\documentclass[a4paper,UKenglish,cleveref,autoref,thm-restate]{lipics-v2021}

\hideLIPIcs  

\title{Reach together: How populations win repeated games}

\author{Nathalie Bertrand}{Univ Rennes, Inria, CNRS, IRISA,
  France}{nathalie.bertrand@inria.fr}{http://orcid.org/0000-0002-9957-5394}{}
\author{Patricia Bouyer}{Université Paris-Saclay, CNRS, ENS
 Paris-Saclay, LMF, 91190 Gif-sur-Yvette,
 France}{bouyer@lmf.cnrs.fr}{https://orcid.org/0000-0002-2823-0911}{}
\author{Luc Lapointe}{Université Paris-Saclay, CNRS, ENS Paris-Saclay,
 LMF, 91190 Gif-sur-Yvette,
 France}{luc.lapointe@lmf.cnrs.fr}{}{}
\author{Corto Mascle}{MPI-SWS, Kaiserslautern,
 Germany}{cmascle@mpi-sws.org}{}{}

\authorrunning{N. Bertrand, P. Bouyer, L. Lapointe and C. Mascle}
\Copyright{Nathalie Bertrand, Patricia Bouyer, Luc Lapointe and Corto Mascle} 



\category{} 

\relatedversion{}

\acknowledgements{This project has received funding from ANR BisoUS, ANR PaVeDyS.}

\nolinenumbers 

\EventEditors{}
\EventNoEds{2}
\EventLongTitle{Submission to STACS 2026}
\EventShortTitle{STACS 2026}
\EventAcronym{STACS}
\EventYear{2026}
\EventDate{}
\EventLocation{}
\EventLogo{}
\SeriesVolume{}
\ArticleNo{}

\usepackage{defs}
\usepackage{multidef}
\usepackage{xspace}
\usepackage{MnSymbol,wasysym} 
\usepackage{stmaryrd} 
\usepackage{todonotes}
\usepackage{wrapfig}

\usepackage[ruled]{algorithm}
\usepackage{algorithmicx}
\usepackage{listings}
\usepackage[noend]{algpseudocode}
\usepackage{multicol}

\usepackage{mathtools}

\usepackage{tikz}
\usetikzlibrary{automata, snakes, positioning, arrows,arrows.meta,calc,decorations,decorations.text,decorations.markings,math,shapes.geometric}
\tikzset{
	>=stealth',
	-={stealth',ultra thick,scale=3} 
	node distance=1cm, 
	every state/.style={thick}, 
	initial text=$ $, 
}
\usepackage{tikzit}

\usetikzlibrary{arrows, arrows.meta,calc,automata,fit,shapes,positioning}
\tikzset{AUT style/.style={>=angle 60,thick,node distance=1.2, initial text= ,every edge/.append,every state/.style={fill=gray!10!white,minimum size=10,inner sep=2}}}

\definecolor{Green2}{HTML}{44CC44}
\definecolor{Red2}{HTML}{FF0400}
\definecolor{Blue2}{HTML}{87CEFA}

\usepackage{multicol}
\usepackage{fontawesome5}

\newcommand{\good}{\colorbox{white}{\color{Green2}\faCheck}}
\newcommand{\bad}{\colorbox{white}{\color{Red2}\faTimes}}
\newcommand{\neutral}{\colorbox{white}{\color{Blue2}\faMinus}}

\newcommand{\etiquette}[1]{\lambda(#1)}
\begin{document}

\maketitle

\begin{abstract}  
  In repeated games, players choose actions concurrently at each
  step. We consider a parameterized setting of repeated games in which
  the players form a population of an arbitrary size. Their utility
  functions encode a reachability objective. The problem is whether
  there exists a uniform coalition strategy for the players so that
  they are sure to win independently of the population size. We use
  algebraic tools to show that the problem can be solved in polynomial
  space. First we exhibit a finite semigroup whose elements summarize
  strategies over a finite interval of population sizes. Then, we
  characterize the existence of winning strategies by the existence of
  particular elements in this semigroup. Finally, we provide a
  matching complexity lower bound, to conclude that repeated
  population games with reachability objectives are \PSPACE-complete.
\end{abstract}

\newpage

\section{Introduction}
\subparagraph*{Games} Game theory is a field that introduces and
studies models of interactions between several agents, also called
players~\cite{OR-book94}. The players are most often supposed to be
rational: their goal is to act so as to maximize their utility. One of
the simplest models of games is the one of two-player zero-sum
games. In these games, as the name suggests, two players interact and
have opposite objectives, which is represented by the fact that for
every play, the sum of their payoff is null. An example of such games
is the one of rock-paper-scissors, in which one can choose the utility
to be $1$ in case of a win, $0$ for a draw and $-1$ for a loss. It
falls in the class of concurrent games, in which players make their
decision simultaneously~\cite{dAH00,dAHK07}. The strategies of the
players can either be pure, that is, each player chooses action
deterministically, or randomized, that is, using a form or another of
randomness (see~\cite{MR-ic24} for a taxonomy of variants in
randomization). A natural solution concept for two-player zero-sum
games is the one of \emph{value}, that is the greatest payoff
Player~$1$ can guarantee independently of the choices of their
opponent. Assuming randomized strategies, in the above
rock-paper-scissors game, it happens to match the lowest payoff
Player~$2$ can ensure, whatever the choices of Player~$1$; the
generalization of this result to stochastic games is known as
Blackwell determinacy~\cite{martin98}.

To model situations with more than two entities, one uses multiplayer
games. In multi-player games, each player has their own utility
function that they try to maximize. In contrast to two-player zero-sum
games, their objectives are not necessarily conflicting, so that the
value is not relevant. Several solution concepts have been defined for
multiplayer concurrent games, such as winning (pure)
strategies~\cite{dAHK07}, rationality of players~\cite{FKL10}), or
Nash equilibria~\cite{Nash50}.
In words, a Nash equilibrium is a contract between the players in the
form of a strategy profile, \emph{i.e.} a strategy for each
player. This profile is such that no individual player has an
incentive to unilaterally deviate from the agreed equilibrium. The
existence of Nash equilibria and their computation is an important
research question for various models of multiplayer
games~\cite{UW11a,BBMU15}.

\subparagraph*{Repeated games}
A repeated game consists of repetitions
of a so-called stage game. The stage game can for instance be in
normal form: each player chooses an action from a finite set and a
matrix gives the payoffs of each player according to their combined
choices. Repeated games are then repetitions of this normal-form game,
and players can take into account the past actions and payoffs to
decide on their current action. On the rock-paper-scissors example,
the stage game is one round of rock-paper-scissors, and the repeated
game is the repetition ad infinitum of the stage game.  The set of
payoffs that are achieved by equilibria can be characterized, be it
among randomized strategies~\cite{Sorin86} or pure
strategies~\cite{Tomala98}. 
Extensions of the framework have been considered, for instance to
incorporate partial observation by the players of the played actions
and stage payoffs at each round~\cite{GRSVZ-corr14}.

Big Match is another classical example of repeated
game~\cite{BF-ams63}, that allows us to introduce the notion of
absorbing payoffs. At every round, Player~$2$ chooses a letter, either
$a$ or $b$, and Player~$1$ tries to guess their choice. If Player~$1$
is correct, they earn a stage payoff of $1$, otherwise the utility is
$0$. This continues until Player~$1$ predicts $a$: from then on, both
players must stick to their decision at that round. In case of a match
at that round, Player~$1$ will earn $1$ in each following round,
otherwise they will forever have payoff $0$. The payoffs are said to
be \emph{absorbing} when the decision of Player~$1$ is $a$. The value
of Big Match is the greatest mean-payoff Player~$1$ can ensure on the
sequence of stage utilities. It happens here also to match the least
mean-payoff on the sequence of utilities Player~$2$ can ensure.

\subparagraph*{Games with arbitrarily many players} In recent years,
several models of games with an arbitrary number of players have been
introduced, to model situations in which the number of agents is
unknown, or concisely represent infinitely many games instances, one
for each number of players. Concurrent parameterized
games~\cite{BBM-fsttcs19,BBM-fsttcs20} and population control
problems~\cite{BDGGG-lmcs19,CFO-lmcs21,GMT-corr25} both
fall in this category. Different to the normal-form games and repeated
games, these are played on a graph --or an automaton-- and one or
several pebbles move along transitions during a play. For population
control, the goal is to drive the pebbles to a common goal,
independently of the population size, \emph{i.e.} for every possible
number of pebbles. In concurrent parameterized games, two main
problems have been considered, depending on whether the players
collaborate to achieve a safety goal, or one player plays against the
coalition of others to achieve a reachability objective. For both
problems, the population size is not known in advance, and players
should have uniform strategies that do not depend on their
number. Population games and parameterized concurrent games both fit
in the framework of parameterized verification, where the parameter is
the number of entities, here the number of players.

\subparagraph*{Contributions} In this paper, we introduce a model that
reconciles multiplayer repeated games and games with arbitrarily many
players. We consider repeated games for populations, with absorbing
payoffs to encode reachability-like objectives. In the stage game,
rather than a single utility function, games for populations are
defined by infinitely many utility functions, one for each population
size, \emph{i.e.} number of players. Towards the definition of a
decision problem, this sequence of utility functions is given by a
deterministic automaton equipped with transition labels that reflect
the utility. The joint moves of the (unboundedly many) players form an
infinite word, that one reads in the automaton, and the successive
labels of transitions determine the payoffs for every population size
one after another.

Rather than adversarial settings or Nash equilibria, we focus here on
coalition strategies with which all players aim at achieving a goal
collectively. The problem we are interested in is, given a labelled
deterministic automaton, whether there exists a uniform coalition
strategy for the population to guarantee, for every population size, a
maximal payoff of $1$, representing that a target has been reached.

We use algebraic techniques to prove that this problem can be solved
in \PSPACE. Recall that a population move is encoded by an infinite
path in the labelled automaton. We define a structure of semigroup in
which the elements, called \emph{frontiers}, summarize the effect of a
sequences of finite portions of such paths. The semigroup internal
operation intuitively corresponds to concatenation of these path
portions. We then define a morphism $\psi$ from $\{0,1\}$ to the set
of frontiers that allows one to describe slices of winning
paths. Thanks to Ramsey's theorem on infinite graphs with coloured
edges, the positive instances of the repeated game for populations can
be characterized using this morphism, in our main technical result:
%
%

%
%
%
%
\begin{theorem}
  \label{th:characterization-positive-instances-with-frontiers}
  There is a winning population strategy if and only if there exist
  two frontiers \(f,g \in \psi(0^+)\) such that:
  \begin{enumerate}
    \begin{minipage}[t]{.2\linewidth}\item \(f\) is initial,
    \end{minipage}
    \begin{minipage}[t]{.25\linewidth}
    \item \(g\) is \(\omega\)-iterable,
    \end{minipage}
    \begin{minipage}[t]{.2\linewidth}
    \item \(f * g \to f\),
    \end{minipage}
    \begin{minipage}[t]{.2\linewidth}
    \item \(g * g \to g\).
    \end{minipage}
  \end{enumerate}
\end{theorem}

The precise notions and notations of this result will be made clear
later in the paper. Intuitively, the condition that $f$ is initial
ensures that it encodes a prefix of all moves, and the
$\omega$-iterability of $g$ guarantees further portions of all moves
follow a regular pattern, and the two other conditions impose that $f$
and $g$ combine nicely. Since the number of frontiers is at most
exponential in the size of the labelled deterministic automaton, this
characterization allows one to decide the problem in polynomial space.
  %

Finally, we provide a matching complexity lower-bound to conclude that
repeated games for populations with reachability objectives are
\PSPACE-complete.

\section{Repeated games for populations}
\label{sec:setting}

\subsection{Notations}
We write $\NN_{>0}$ for the set of positive integers. Given
$i \le j \in \NN_{>0}$, $\llbracket i;j \rrbracket$ denotes the set
$\set{i, i+1, \ldots, j}$.  We write $\Sigma^*$ for the set of finite
words over an alphabet $\Sigma$, and $\Sigma^\omega$ for the set of
infinite words.  The length of a finite word $w$ is denoted $|w|$, and
if $w$ is infinite, we write $|w|= \infty$.  Given a finite or
infinite word \(w\) and \(n \in \NN_{>0}\), we denote by \(w[n]\) its
\(n\)-th letter, \(w[:n]\) its prefix of length \(n\), \(w[n:]\) the
suffix obtained by removing its first \(n-1\) letters, and more
generally $w[n:m]$ for its factor starting at position $n$ up to
position $m$. If $(i_j)_{j \in J}$ is a sequence of indices, we write
$w[(i_j)_{j \in J}]$ for the sequence $(w[i_j])_{j \in J}$.
In particular, $w[n,m]$ is the pair $(w[n], w[m])$.

A deterministic finite automaton is a tuple
$\mathcal{A} = (Q, q_{\mathrm{init}}, \Sigma, \delta)$ with $Q$ a set
of states, $q_{\mathrm{init}}$ an initial state, $\Sigma$ the alphabet
and $\delta : Q \times \Sigma \to Q$ a partial function. We assume the
reader is familiar with basic automata theory. 

\subsection{Description of the setting}
\label{sec:setting}
\subparagraph*{Repeated games for populations} Extending the framework
of~\cite{GRSVZ-corr14} to multiple players, we define the notion of
repeated games with $N$ players.
\begin{definition}
  Let $\Sigma$ be an alphabet and $N \in \NN_{>0}$. A \emph{repeated
    game with $N$ players and absorbance} is given by a stage utility
  function $u: \Sigma^N \to [-1;1] \cup ([-1;1] \times
  \{\bullet\})$. 
\end{definition}
The notation $\bullet$ is for absorbing payoffs, and an element
$(p,\bullet) \in [-1;1] \times \{\bullet\}$ will simply be noted
$p^\bullet$ in the following.
In such a game, a \emph{move} is a word $\move \in \Sigma^N$; for
every $n \le N$, the $n$-th letter of $\move$ corresponds to the
action played by player $n$.  A \emph{coalition strategy} is an
infinite sequence of moves $\sigma = (\move_i)_{i \in \NN_{>0}}$: it
generates the infinite sequence of payoffs
$(u(\move_i))_{i \in \NN_{>0}}$, which aggregates (i) to $p_{i_0}$ if
$u(\move_{i_0}) = p_{i_0}^\bullet$ with $i_0$ the least index $i$
such that $u(\move_i) \in [-1;1] \times \{\bullet\}$; or (ii) to
$\limsup_{n \to +\infty} \frac{\sum_{i=1}^n u(\move_i)}{n+1}$
otherwise. In the special case where the codomain of $u$ is
$\{-1^\bullet,0,1^\bullet\}$, then the aggregate payoff is either $-1$
or $1$ (case (i)), or $0$ (case (ii)), and we speak of
\emph{reachability payoff}.

  

Focusing on reachability payoffs, we now define repeated games for
populations as an extension of the previous model to arbitrarily many
players.
\begin{definition}
  \label{def:param-repeated-game}
  Let $\Sigma$ be an alphabet.  A \emph{reachability repeated game for
    populations}, or simply \emph{population game}, is given by, for every
  $N \in \NN_{>0}$ a utility function
  $u_N: \Sigma^N \to \{-1^\bullet,0,1^\bullet\}$.
\end{definition}
Since the number of players is arbitrary in repeated games for
populations, a \emph{move} is an infinite word
$\move \in \Sigma^\omega$. It is played uniformly on all games with
a fixed number of players, resulting in one stage utility
$u_N(\move[:N]) \in \{-1^\bullet,0,1^\bullet\}$ for every possible
number $N$ of players.
A \emph{population strategy} $\sigma$ is an infinite sequence of moves
$(\move_i)_{i \in \NN_{>0}}$. Such a strategy yields an aggregate
payoff vector $P_\sigma = (p_{\sigma,N})_{N \in \NN_{>0}}$, where
$p_{\sigma,N}$ is the aggregate payoff of the repeated game with stage
utility $u_N$, corresponding to $N$ players. The population strategy
$\sigma$ is \emph{winning} whenever for every $N \in \NN_{>0}$,
$p_{\sigma,N} = 1$.



\subparagraph*{Automata-defined population games} We consider reachability
repeated games for populations in which the family of utility
functions is presented as a pair $\langle \calA,\lambda \rangle$
formed of a \emph{deterministic finite-state automaton}
$\calA = (Q, q_{\text{init}}, \Sigma, \delta)$ together with a function
$\lambda : Q \times \Sigma \to \{\good,\neutral, \bad\}$
labeling 
transitions as ``good'' ($\good$), ``neutral'' ($\neutral$) or ``bad''
($\bad$). Given a finite word $u \in \Sigma^+$, we write
$\etiquette{q, u} \in \{\good,\neutral,\bad\}$ for the label of the
last transition taken in $\langle \calA,\lambda \rangle$ upon reading $u$ from state $q$ in
$\calA$, when defined; if $q = q_{\mathrm{init}}$ is the initial state
of $\calA$, we simply write $\etiquette{u}$. Such a \emph{labelled
  automaton} $\langle \calA,\lambda \rangle$ thus classifies non-empty
words between $\good$, $\neutral$ and $\bad$.

Note that $\langle \calA,\lambda \rangle$ defines all utility functions at once, as we shall
see now.  The stage payoff of move $\move \in \Sigma^\omega$ in the
game with $N$ players is defined as 
\[ u_N(\move[:N]) =
  \left\{\begin{array}{@{}cl}
           -1^\bullet & \text{if}\ \etiquette{\move[:N]} = \bad\\
    0 & \text{if}\ \etiquette{\move[:N]} = \neutral\\
    \phantom{-}1^\bullet & \text{if}\ \etiquette{\move[:N]} = \good
           \end{array}\right.
\]
               We write $\game{\langle \calA,\lambda \rangle}$ for the
               induced repeated game for populations.


\subparagraph*{Problem definition} We are interested in winning
population strategies, that is, in population strategies that achieve
an aggregate payoff of $1$ for every possible number of players.
The population strategy $\sigma = (\move_i)_{i \in \NN_{>0}}$ is
\emph{winning} if for every $N\in \NN_{>0}$ there exists an integer
$i(N)$ such that $\etiquette{\move_{i(N)}[:N]} = \good$ and for every
$i<i(N)$, $\etiquette{\move_i[:N]} = \neutral$. The terminology of
``reachability'' should now become clear: $\good$ (resp. $\bad$) are
labels of accepting (resp. rejecting) transitions, while $\neutral$
means undecided. To win for some population size $N$, label $\good$
needs to appear exactly at that position, while no $\bad$ has been
previously encountered at that position, thus resembling a constrained
reachability property. The terminology will further be justified in
Section~\ref{subsec:relationship}, when we will give the relationship
with parameterized concurrent games.

More permissively, the population strategy
$\sigma = (\move_i)_{i \in \mathbb{N}}$ is \emph{non-losing} if for
every $N\in \NN_{>0}$, whenever there exists an integer $i(N)$ such
that $\etiquette{\move_{i(N)}[:N]} = \bad$, then there exists
$i<i(N)$, $\etiquette{\move_i[:N]} = \good$. Note that a winning
strategy is non-losing; however a non-losing strategy might be non
winning, if for some $N \in\nats_{>0}$, for every $i$,
$\etiquette{\move_i[:N]} = \neutral$.


We are now in a position to formally state our decision problem:

\medskip\noindent\fbox{\begin{minipage}{.99\linewidth}
    \underline{{\reachtogether}} \\
    {\bf Input}: A deterministic finite automaton $\calA$
    with transition labeling $\lambda :\delta \to \{\good,\neutral,\bad\}$.\\
    {\bf Question}: Does there exist a winning population strategy in
    $\game{\langle \calA,\lambda\rangle}$?
  \end{minipage}}\medskip

Our main contribution is to establish that the above decision problem
is \PSPACE-complete.

\begin{remark}
  Since $\calA$ is deterministic, choosing a word in $\Sigma^*$ is the
  same as choosing a path in $\calA$. A winning strategy is thus
  simply a sequence of infinite paths in $\calA$ such that, for all
  $N \in \NN_{>0}$ there is a path whose $N$-th transition is labelled
  $\good$ and all previous paths have their $N$-th transition is
  labelled by $\neutral$.  As a consequence, in all the forthcoming
  figures, we omit the letters in the automata, and only indicate the
  labels.
\end{remark}

\section{Playing repeated population games}
\subsection{Examples}
To get familiar with the model of
reachability repeated games for populations, we provide a couple of
examples. 

\begin{example}
  \label{ex1}
  Let us examine the labelled automaton in
  Figure~\ref{fig:ex1bis}. A winning population strategy is
  depicted on Figure~\ref{fig:possible-visited-statesbis}. It first
  selects the path that loops $0$ times on the initial state, then the
  path that loops once, then twice, and so on. This permits to win for
  every population size one by one in increasing order.
\end{example}

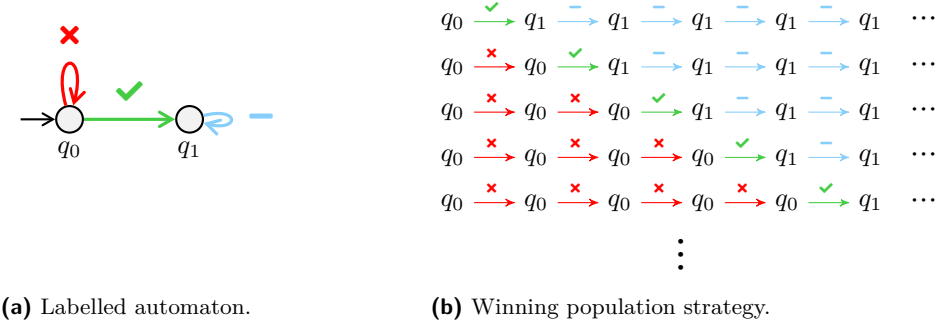
\begin{figure}[H]
  \hfill
	     \begin{subfigure}{.3\textwidth}
               \begin{tikzpicture}[AUT style]
	
	\node[state, initial] (A) {};
	\node[state, right= of A] (B)  {};
	
	\node[yshift=-4mm] (At) at (A) {$q_0$};
	\node[yshift=-4mm] (Bt) at (B)  {$q_1$};
	
	\draw[->, loop above, color=Red2, very thick, looseness=20] (A) edge[above] node {$\bad$} (A);
	\draw[->] (A) edge[above,  very thick, color=Green2] node {$\good$} (B);
	\draw[->, loop right] (B) edge[right, very thick, color=Blue2] node {$\neutral$} (B);
	
	\node[below=of B, yshift=-5mm] (C) {}; 
\end{tikzpicture}
               \caption{Labelled automaton.}\label{fig:ex1bis}
             \end{subfigure}
             \hfill
	     \begin{subfigure}{.5\textwidth}
		       \begin{tikzpicture}
	\foreach \x in {0,...,5}
	{
		\foreach \y  in {0,...,4}
		{
			\pgfmathparse{\x+\y-5 < 0 ? 1 : 0}
			\ifthenelse {\pgfmathresult>0}{
				\node (\x\y) at (1.1*\x,0.6*\y) {$q_0$};
				\ifthenelse{\x>0}
				{
					\pgfmathparse{int(\x-1)}
					\draw[->, color=Red2] (\pgfmathresult\y) edge node[above, yshift=-3pt,xshift=-1pt] {\tiny$\bad$} (\x\y);	
				}{};
			}{};
			
			\pgfmathparse{\x+\y-5 == 0 ? 1 : 0}
			\ifthenelse {\pgfmathresult>0}{
				\node (\x\y) at (1.1*\x,0.6*\y) {$q_1$};
				\pgfmathparse{int(\x-1)}
				\draw[->, color=Green2] (\pgfmathresult\y)  edge node[above, yshift=-3pt,xshift=-1pt] {\tiny$\good$} (\x\y);
			}{};	
		
			\pgfmathparse{\x+\y-5 > 0 ? 1 : 0}
			\ifthenelse {\pgfmathresult>0}{
				\node (\x\y) at (1.1*\x,0.6*\y) {$q_1$};
				\pgfmathparse{int(\x-1)}
				\draw[->, color=Blue2] (\pgfmathresult\y)  edge node[above, yshift=-3pt,xshift=-1pt] {\tiny$\neutral$} (\x\y);
			}{};
		}
	}

	\node[below= of 32, yshift=-3mm, xshift=-3mm] (vd) {\LARGE $\vdots$};
	
	\foreach \y  in {0,...,4}
	{
		\node (dots\y) at (6.2, 0.6*\y) {\Large $\cdots$};
	}
      \end{tikzpicture}
		       \caption{Winning population strategy.}\label{fig:possible-visited-statesbis}
		     \end{subfigure}
                     \caption{Example of a labelled deterministic
                       automaton $\langle \calA,\lambda \rangle$
                       (left) for which there is a winning population
                       strategy in
                       $\calG(\langle \calA,\lambda \rangle)$
                       (right).}
                     \label{fig:Ex1}
             \end{figure}

\begin{example}
  We now consider the example in Figure~\ref{fig:ex3}, with parameters
  $n_1, n_2 \in \NN_{>0}$.  The upper branch permits to win for every
  population size $k{+}1$ with $k \equiv 0 \mod n_1$.  The lower
  branch permits to win for a number of players that is $n_2$ less
  than a population size that has already won. Combining those two
  types of moves, one can win for every $N$ of the form
  $1+ \alpha_1 n_1 - \alpha_2 n_2$ with $\alpha_1, \alpha_2 \in
  \NN$. This covers all positive positions if and only if $n_1$ and
  $n_2$ are co-primes, by Bézout's theorem.
  \begin{figure}[ht]
		\begin{tikzpicture}[AUT style]
	
	\node[state, initial] (A) {};
	
	\node[state ,above right= of A,yshift=-5mm] (B1)  {};

        \node[state,right= of B1] (B2) {};
	\node[right = of B2, xshift=-9mm] (B3) {\Large \textbf{...}};
	\node[state,right = of B3] (B4)  {};
	
	\draw[->, bend left=30,color=Green2] (A) edge[above left] node[xshift=3pt,yshift=-3pt] {$\good$} (B1);
	\draw[->,color=Blue2] (B1) edge[above] node[yshift=-3pt] {$\neutral$} (B2);
        \draw[->,color=Blue2] (B3) edge[above] node[yshift=-3pt] {$\neutral$} (B4);
	\draw[->, bend left=20, color=Green2] (B4) edge[below] node[yshift=3pt] {$\good$} (B1);

        \draw [decorate,decoration={brace,amplitude=5pt, raise=4ex}]
	(1,.5) -- (5.3,.5) node[midway,yshift=3em]{$n_1$ states};

	\node[state,below right= of A,yshift=5mm] (D1)  {};
	\node[state, right= of D1] (D2) {};
	\node[state, right = of D2] (D3)  {};
	\node[ right= of D3, xshift=-9mm] (U) {\Large \textbf{...}};
	\node[state, right= of U] (V) {};
	\node[state, right= of V] (D4) {};

	\draw[->,bend right =30,color=Blue2] (A) edge[below left] node[xshift=3pt,yshift=3pt] {$\neutral$} (D1);	
	\draw[->, loop below,color=Blue2] (D1) edge[below] node {$\neutral$} ();
	\draw[->, color=Green2] (D1) edge[below] node[yshift=3pt] {$\good$} (D2);
	\draw[->,color=Blue2] (D2) edge[below] node {$\neutral$} (D3);
        \draw[->] (U) edge[below,color=Blue2] node {$\neutral$} (V);        
        \draw[->,color=Red2] (V) edge[below] node[yshift=3pt] {$\bad$} (D4);
	\draw[->, loop right,color=Blue2] (D4) edge[right] node[xshift=-3pt,yshift=-1pt] {$\neutral$} ();

	\draw [decorate,decoration={brace,amplitude=5pt,mirror, raise=4ex}]
	(2.6,-.5) -- (6.85,-.5) node[midway,yshift=-3em]{$n_2$ states};
	
	
\end{tikzpicture}
		\caption{Example of a labelled automaton which is a
                  positive instance of \reachtogether if and only if
                  $n_1$ and $n_2$ are co-primes.}
		\label{fig:ex3}
	\end{figure}
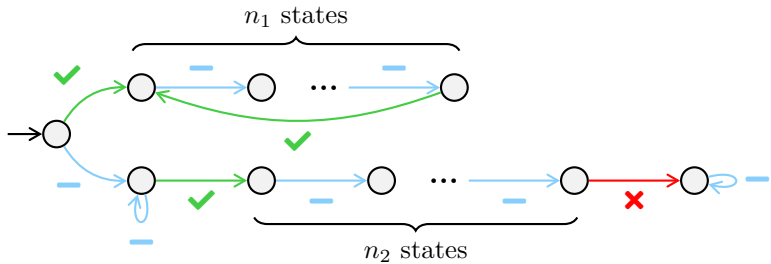
\end{example}

\subsection{Further example and first complexity lower bound}
One can already state that \reachtogether is \coNP-hard, via a
reduction from the universality of unary non-deterministic automata.
We only sketch the proof here, since we will later show that the
problem is in fact \PSPACE-hard.

Take a non-deterministic finite automaton (NFA in short) $\mathcal{N}$
over a unary alphabet $\set{a}$ (we can assume without loss of
generality that it has a single initial state).  Relabel transitions with fresh letters
so that $\mathcal{N}$
becomes deterministic.  Add a state $s_\top$, transitions labelled
$\good$ from every final state of $\mathcal{N}$ to $s_\top$, and a
loop labelled $\neutral$ on $s_\top$. Finally, label all transitions
in $\mathcal{N}$ with $\neutral$. The obtained labelled automaton is
depicted in Figure~\ref{fig:coNP}.
  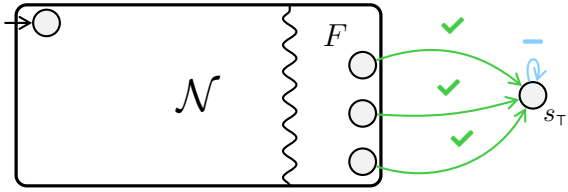
\begin{figure}[h]
    \qquad\begin{tikzpicture}[AUT style,scale=.8]
	
	\draw[very thick, rounded corners] (0,0) rectangle (6,3);
	\draw[decorate, decoration={snake}] (4.5, 0) -- (4.5,3);

	\node[state, initial] (A) at (0.5,2.7) {};
	\node at (3,1.5) {\LARGE $\mathcal{N}$};
	\node at (5.25,2.5) {\Large $F$};

	\node[state] (F1) at (5.7, 2)  {};
	\node[state] (F2) at (5.7, 1.2)  {};
	\node[state] (F3) at (5.7, 0.4)  {};
	
	\node[state] (G) at (8.5, 1.5)  {};
        \node[xshift=3mm,yshift=-3mm] (At) at (G) {$s_\top$};
	
	\draw[->, bend left,color=Green2] (F1) edge node[above] {\footnotesize $\good$} (G);
	\draw[->, bend right=10,color=Green2] (F2) edge node[above] {\footnotesize $\good$} (G);
	\draw[->, bend right=40,color=Green2] (F3) edge node[above] {\footnotesize $\good$} (G);
	\draw[->, loop above, color=Blue2] (G) edge node[above,yshift=-3pt] {\footnotesize $\neutral$} ();

\end{tikzpicture}
    \caption{A construction to prove \coNP-hardness.}
    \label{fig:coNP}
  \end{figure}

  Observe that there is a winning strategy for the population game
  defined by resulting labelled automaton if and only if there are
  paths from the initial state to $s_\top$ of every positive length,
  if and only if $\mathcal{N}$ accepts all words over $a^*$.
	
  Since universality of unary NFAs is \coNP-complete~\cite[Theorem
  6.1]{StockmeyerM73} \reachtogether\ is \coNP-hard.  In the case
  where there are no transitions labelled $\bad$, it is even
  \coNP-complete. Indeed the problem boils down to finding paths of
  every positive length ending with a $\good$-labelled
  transition. This in turn reduces to universality of a unary NFA by
  making every source of a $\good$-labelled transition a final state.

  \subsection{Relationship with parameterized concurrent games}
  \label{subsec:relationship}
  Repeated games for populations 
  can be interpreted as concurrent parameterized games, as first
  investigated in~\cite{BBM-fsttcs19,BBM-fsttcs20}. We describe the
  model on the example from Figure~\ref{fig:anirban1}. When playing a
  concurrent parameterized game, the number of players is finite but
  arbitrary. Here, the objective is assumed to be a reachability
  objective, and the goal is to reach the target state ($v_1$ here)
  from the initial state $v_0$, whatever the number of players.  A
  move is an infinite word, for instance $a^2ba^\omega$, where the
  prefix of length $k$ represents the move of the players in case
  there are $k$ of them.  With that move, from $v_0$, if there are three players,
  the game will proceed to the target (since $a^2b \in a^*b\Sigma^+$);
  if there are two players, the game will proceed to $v_2$ (since
  $a^2 \in a^+$). On this instance, there exists a winning coalition
  strategy, which consists in playing from the initial state
  $b a^\omega$, then
  $aba^\omega$, then $b^2 a b^\omega$, etc. Under this strategy
  --uniform for every possible number of players-- if there are $k$
  players, the game will loop $(k-1)$ times in $v_0$ and then proceed
  to the target
  $v_1$. 

  \begin{figure}[h]
    \begin{center}
      \hfill
      \begin{subfigure}{.45\textwidth}
        \begin{center}
        \begin{tikzpicture}[AUT style]
          
          \node[state, initial, shape=rectangle] (A) {};
          \node[state, right= of A,double, shape=rectangle] (B)  {};
          \node[state, below= of A,shape=rectangle] (C) {};
          
          \node[yshift=-4mm,xshift=4mm] (At) at (A) {$v_0$};
          \node[yshift=-5mm] (Bt) at (B)  {$v_1$};
          \node[xshift=-5mm] (Ct) at (C) {$v_2$};
          
          \draw[->, loop above,looseness=20] (A) edge[above] node {$a^*b\Sigma^+$} (A);
          \draw[->] (A) edge[above] node {$a^*b$} (B);
          \draw[->] (A) edge[left] node {$a^+$} (C); 
          \draw[->, loop right] (B) edge[right] node {$\Sigma^+$} (B); 
          \draw[->, loop right] (C) edge[right] node {$\Sigma^+$} (C);
          
          \node[below=of B, yshift=-5mm] (C) {}; 
        \end{tikzpicture}
        \end{center}
        \caption{An example arena with a reachability objective
          (target state $v_1$), where there is a  coalition winning
          strategy.}
        \label{fig:anirban1}
      \end{subfigure}
      \hfill
      \begin{subfigure}{.45\textwidth}
        \begin{center}
        \begin{tikzpicture}[AUT style]
          
          \node[state, initial,rectangle] (A) {};
          \node[state, right= of A,double,rectangle] (B)  {};
          \node[state, below= of A,rectangle] (C) {};
          
          \node[yshift=-5mm,xshift=5mm] (At) at (A) {$v_{\small\neutral}$};
          \node[yshift=-6mm,xshift=3mm] (Bt) at (B)  {$v_{\small\good}$};
          \node[xshift=-5mm,yshift=-3mm] (Ct) at (C) {$v_{\small\bad}$};
          
          \draw[->, loop above,looseness=20] (A) edge[above] node {$\calL_{\small\neutral}$} (A);
          \draw[->] (A) edge[above] node {$\calL_{\small\good}$} (B);
          \draw[->] (A) edge[left] node {$\calL_{\small\bad}$} (C); 
          \draw[->, loop right] (B) edge[right] node {$\Sigma^+$} (B); 
          \draw[->, loop right] (C) edge[right] node {$\Sigma^+$} (C);
          
        \end{tikzpicture}
        \end{center}
        \caption{Reachability population game seen as a coalition
          parameterized game. \\ }
        \label{fig:reach}
      \end{subfigure}
      \caption{Link with concurrent parameterized
        games~\cite{BBM-fsttcs20} with a reachability objective.}
      \label{fig:jeu-anirban}
    \end{center}
\end{figure}
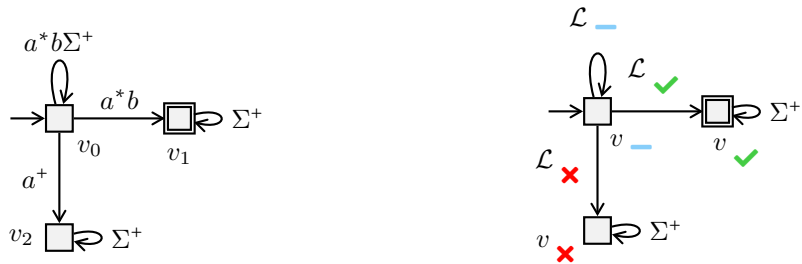


Let us explicit the relationship between repeated games for
populations and concurrent parameterized games. From a labelled
automaton $\langle \calA, \lambda\rangle$ one builds a concurrent
parameterized arena with three vertices
$\{v_{\scriptsize \neutral},v_{\scriptsize\bad},v_{\scriptsize\good}\}$ as in Figure~\ref{fig:reach},
and one derives the regular languages labelling the transitions of the
arena as follows: language $\calL_{\scriptsize\neutral}$ (resp. $\calL_{\scriptsize\good}$,
resp. $\calL_{\scriptsize\bad}$) is the language of finite words over $\Sigma$,
whose execution in $\calA$ ends with a $\neutral$-labelled
(resp. $\good$-labelled, resp.  $\bad$-labelled) transition. The arena
of Figure~\ref{fig:anirban1} is obtained after applying this
construction to the labelled automaton of Figure~\ref{fig:ex1bis} (in
which from $q_0$, letter $a$ is associated with the transition
labelled by $\bad$ while letter $b$ is associated with the transition
labelled by $\good$). Under that construction, there exists a winning
population strategy in the repeated game defined by
$\langle \calA,\lambda \rangle$ if and only if there exists a strategy
that ensures that for every number of players, $v_{\scriptsize\good}$ is
reached. Reciprocally, from a concurrent parameterized game with a
three-state shape as on Figure~\ref{fig:anirban1}, one can build a
labelled automaton such that solving the population game on the latter
will solve the coalition problem on the former. To the best of our
knowledge, the reachability problem for coalitions in concurrent
parameterized games (over an arbitrary finite-graph-based arena)) is
an open problem since~\cite{BBM-fsttcs19,BBM-fsttcs20}. The current
paper thus presents a solution to that open problem on specific
instances.


\section{Deciding how populations uniformly win}
\label{sec:deciding}
We fix a labelled automaton $\langle \calA ,\lambda \rangle$ with
$\calA = (Q, q_{\text{init}}, \Sigma, \delta)$ for the rest of this
section.

A population strategy is an infinite sequence of moves, each inducing
an infinite path in $\calA$.  To exhibit a winning strategy for
positive instances of \reachtogether, one needs to be able to express
the full sequence of paths with a finite description, which we do by
identifying some repeating pattern.
To do so, we change our point of view: instead of a sequence of
infinite paths in the automaton, we see strategies as infinite words
of sequences of transitions.  If we represent strategies as grids as
in Figure~\ref{fig:possible-visited-statesbis}, this means that we
switch our vision of the strategy from an infinite sequence of rows to
an infinite sequence of columns. We will prove that for positive
instances, there always exists a winning strategy that can be
decomposed as in Figure~\ref{fig:schema-global}, with a prefix $f$
followed by iterations of a factor $g$.

\begin{figure}[htbp]
  \centering
\begin{tikzpicture}

  \node at (1.2,3.5) {$f$};
  \node at (4.2,3.5) {$g$};
  \node at (7.8,3.5) {$g$};
  \node at (11.2,3.5) {$g$};
    \foreach \y  in {0,...,4} 
	{
	    \node (0\y) at (0,0.6*\y) {$\bullet$};
	}
	\foreach \y  in {0,...,4} 
	{
	    \node (1\y) at (1.2,0.6*\y) {{\color{lightgray}{$\bullet$}}};
	}
	\foreach \y  in {0,...,4} 
	{
	    \node (2\y) at (2.4,0.6*\y) {$\bullet$};
	}
	\foreach \y  in {0,...,4} 
	{
	    \node (3\y) at (3.6,0.6*\y) {{\color{lightgray}{$\bullet$}}};
	}
	\foreach \y  in {0,...,4} 
	{
	    \node (4\y) at (4.8,0.6*\y) {{\color{lightgray}{$\bullet$}}};
	}
	\foreach \y  in {0,...,4} 
	{
	    \node (5\y) at (6,0.6*\y) {$\bullet$};
	}
	\foreach \y  in {0,...,4} 
	{
	    \node (6\y) at (7.2,0.6*\y) {{\color{lightgray}{$\bullet$}}};
	}
	\foreach \y  in {0,...,4} 
	{
	    \node (7\y) at (8.4,0.6*\y) {{\color{lightgray}{$\bullet$}}};
	}
	\foreach \y  in {0,...,4} 
	{
	    \node (8\y) at (9.6,0.6*\y) {$\bullet$};
	}
	\foreach \y  in {0,...,4} 
	{
	    \node (9\y) at (10.8,0.6*\y) {{\color{lightgray}{$\bullet$}}};
	}

	\foreach \x in {1,...,9}
	{
        \foreach \y  in {0,...,4}
	    {
            \pgfmathparse{int(\x-1)}
		    \draw[->, lightgray] (\pgfmathresult\y) edge (\x\y);
	    }
		
	}
	
	\node[below= of 52, yshift=-5mm, xshift=-3mm] (vd) {\LARGE $\vdots$};
	
	\foreach \y  in {0,...,4}
	{
		\node (dots\y) at (12, 0.6*\y) {\Large $\cdots$};
	}

	\node (110) at (13.2,2.4) {};
	\draw[bend left,densely dotted] (04.north west) edge node[above] {} (24.north east);
    \draw[densely dotted] (00.south west) -- (04.north west);
    \draw[densely dotted] (20.south east) -- (24.north east);
	\draw[bend left,dashed] (24.north west) edge node[above] {} (54.north east);
    \draw[dashed] (20.south west) -- (24.north west);
    \draw[dashed] (50.south east) -- (54.north east);
	\draw[bend left,dashed] (54.north west) edge node[above] {} (84.north east);
    \draw[dashed] (50.south west) -- (54.north west);
    \draw[dashed] (80.south east) -- (84.north east);
    \draw[bend left,dashed] (84.north west) edge node[above] {} (110.north);
    \draw[dashed] (80.south west) -- (84.north west);
\end{tikzpicture}  
  \caption{Schematic representation of a winning strategy decomposition.}
  \label{fig:schema-global}
\end{figure}
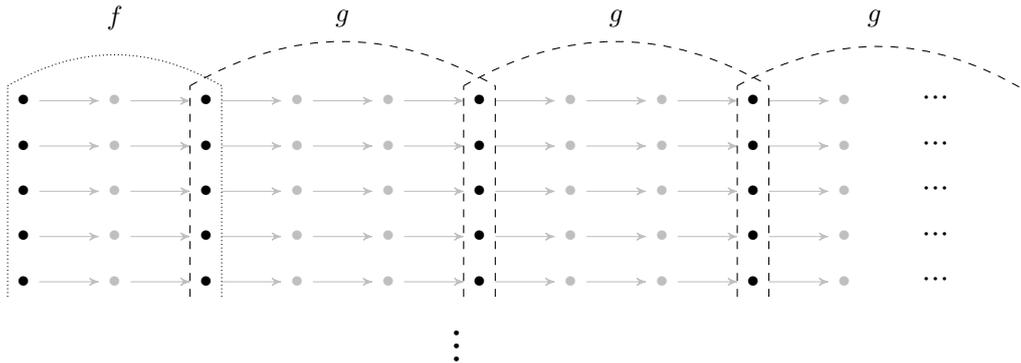

More precisely we determine the existence of winning strategies using
a finite semigroup in which we map sequences of transitions. A
vertical chunk 
is abstracted by the sequence of endpoints of its transitions. To keep
the semigroup finite, repetitions are removed. In other words,
we only keep the set of pairs of states that appear as endpoints of
some element in the chunk, in order of appearance. $f$ and $g$ in the
above figure yield one element each of the semigroup. The internal
operation of the semigroup abstracts the concatenation operation on
vertical chunks.

We then characterise utility functions for which winning population
strategies exist using this semigroup. As announced in
Theorem~\ref{th:characterization-positive-instances-with-frontiers},
positive instances are characterized by the existence of two elements
$f$ and $g$ with four properties: 
(1) $f$ is \emph{initial}, expressing that it indeed corresponds to
the prefix of all moves; (2) $g$ is $\omega$-iterable, expressing that
it describes progress from left to right in the grid; (3) and (4) $f$
and $g$ combine nicely. A key element of the proof of the
left-to-right implication is to apply Ramsey's theorem to an arbitrary
winning population strategy in order to be able to abstract it as a
finite prefix followed by infinitely many repetitions of the same
pattern.

\subsection{Wins words: an alternative view on winning strategies}

We propose an alternative view on winning population strategies, that
will be useful for our proofs.
For a given sequence of moves consistent with a non-losing strategy
\(\sigma\), the set of number of players for which it has already won
at a given time of the sequence is growing as more and more moves are
played.  One can represent this growing set of achieved wins by a
sequence of words in \(\{0, 1\}^\omega\) where the \(j\)-th letter of
the \(i\)-th word is $1$ if
the game is won after the \(i\)-th move when \(j\) players are
involved. Definition~\ref{def:wins-words} formalises this idea.

\begin{definition}
  \label{def:wins-words}
	A \emph{wins word} is an element of $\set{0,1}^+ \cup
        \set{0,1}^\omega$.
	%
	Given two wins words \(w, w' \in \{0,1\}^\ell\) with
        $\ell \in \NN_{>0} \cup \{\omega\}$, a state $q \in Q$ and
        \(\move \in \Sigma^\ell\), we write
        \(w \overset{q \cdot \move}{\rightsquigarrow} w'\) if for all
        \(j > 0\), we have either
        \begin{itemize}
        \item \(w'[j] = w[j] = 1\), or
        \item \(\etiquette{q \cdot \move[:j]} = \neutral\) and
          \(w'[j] = w[j]\), or
        \item \(\etiquette{q \cdot \move[:j]} =\good\) and \(w'[j] = 1\).
        \end{itemize}
        If $q= q_{\mathrm{init}}$ is the initial state of $\calA$, we
        simply write \(w \overset{\move}{\rightsquigarrow} w'\).  We
        write \(w \rightsquigarrow w'\) when there is \(\move\)
        such that \(w \overset{\move}{\rightsquigarrow} w'\) and
        \(\rightsquigarrow^*\) for the reflexive transitive closure of
        \(\rightsquigarrow\).  We define the partial ordering
        $\preceq$ on wins words as: $w \preceq w'$ if and only if
        $|w| = |w'|$ and $w[j] \leq w'[j]$ for all $j$.
\end{definition}

Note the three constraints on evolutions of wins words in the above
definition, that encode the possible successive stage payoffs along a
winning play of the repeated game for a fixed number of players.
Repeated games for populations can be phrased with wins words as
follows: one starts from $0^\omega$ and the goal is to read infinite
words (yielding infinite sequences of labels in
$\langle\calA,\lambda\rangle$) such that one never sees label $\bad$
on a position marked $0$, and makes sure that every position of the
wins word is eventually turned to $1$ with a $\good$ label.

\begin{example}
  Back to Example~\ref{ex1} depicted on
  Figure~\ref{fig:possible-visited-statesbis}. Writing
  $(\move_i)_{i \in \NN_{>0}}$ for the mentioned winning population
  strategy, its corresponding sequence of wins words, starting at
  $0^\omega$ is:
  \[
    0^\omega \overset{\move_1}{\rightsquigarrow} 1 0^\omega
    \overset{\move_2}{\rightsquigarrow} 1 1 0^\omega
    \overset{\move_3}{\rightsquigarrow} 1 1 1 0^\omega
    \overset{\move_4}{\rightsquigarrow} \cdots
  \]
\end{example}

From the above definitions, we immediately derive important properties
of wins words:
\begin{lemma}
  \label{lem:wins-words}
  \begin{itemize}
  \item \textsf{(Winning strategies as wins
      words)} \label{rem:win-strat-wins-words} A strategy
    $\sigma = (\move_i)_{i \in \NN}$ is winning if and only if the
    sequence
    $w_0 \overset{\move_0}{\rightsquigarrow} w_1
    \overset{\mu_1}{\rightsquigarrow} \cdots$ with $w_0 = 0^\omega$ is
    well-defined (that is, the strategy is non-losing) and the
    sequence $(w_i)_{i \in \NN}$ converges to $1^\omega$.    
  \item \textsf{(Monotonicity)} \label{rem:monotonicity} For all
    $w \overset{q\cdot \mu}{\rightsquigarrow} w'$, we have
    $w \preceq w'$. Furthermore, for all $\overline{w}$ such that
    $w \preceq \overline{w}$, we have
    $\overline{w} \overset{q\cdot \mu}{\rightsquigarrow}
    \overline{w'}$ for some $\overline{w'}$ such that
    $w' \preceq \overline{w'}$.
  \item \textsf{(Composition)} \label{rem:composition} If $w$ is
    finite, $w \overset{q\cdot \mu}{\rightsquigarrow} w'$ and
    $\overline{w} \overset{\overline{q}\cdot
      \overline{\mu}}{\rightsquigarrow} \overline{w'}$ with
    $\overline{q} = \delta(q, \mu)$ then
    $w\overline{w} \overset{q\cdot
      \mu\overline{\mu}}{\rightsquigarrow} w'\overline{w'}$.
  \item \textsf{(Invariance under repetitions)} \label{rem:repetitions}
    If $w \overset{q\cdot \mu}{\rightsquigarrow} w'$ then for all
    $w''$ with $w' \preceq w''$ we have
    $w'' \overset{q\cdot \mu}{\rightsquigarrow} w''$.
  \end{itemize}
\end{lemma}

\subsection{Frontiers as summaries of strategies}

We define frontiers, a tool that allows to summarize parts of the
2-dimensional infinite grid representing a strategy as on
Figure~\ref{fig:possible-visited-statesbis}.


\newcommand{\reduction}[1]{\left\langle #1 \right\rangle}

\begin{definition}
  A \emph{frontier} is a sequence of pairs of states of $Q$ of the
  form \({(x_1, y_1), \ldots, (x_p, y_p)}\) without repetitions, i.e.,
  for all $i < j$, $(x_i, y_i) \neq (x_j, y_j)$. The set of frontiers
  is denoted \(\F\).
\end{definition}	
For readability, as a frontier is meant to be a summary of columns of
a table as in Figure~\ref{fig:possible-visited-statesbis}, we will
often write such a sequence of pairs
vertically: \(\seq{x_1,y_1 \\ \vdots \\ x_p,y_p}\). We will do so more
generally for sequences of tuples of states.

For any (finite or infinite) sequence \(\sigma = s_1, s_2, \cdots\) of
elements from a finite set \(S\), we write \(\reduction{\sigma}\) for
the finite sequence obtained from $\sigma$ by removing duplicates,
while keeping the order of first appearance. Observe that for any
sequence of pairs of states $(x_1,y_1), (x_2,y_2), \cdots$, its
\emph{reduction} $\reduction{(x_1,y_1), (x_2,y_2), \cdots}$ is a
frontier.
	
\begin{definition}
  Given three frontiers \(f, g\) and \(h\), we write
  \(f \star g \to h\) if there exist triples of states
  \((x_1, y_1, z_1), \ldots, (x_p, y_p, z_p) \in Q^3\) such that
  \[
    f= \reduction{\seq{x_1,y_1 \\ \vdots \\x_p,y_p}} ; \qquad g =
    \reduction{\seq{y_1,z_1 \\ \vdots \\y_p,z_p}} ; \qquad \textrm{ and } h =
    \reduction{\seq{x_1,z_1 \\ \vdots \\x_p,z_p}}.
  \]
\end{definition}

A frontier can thus be used to summarize the effect of a sequence of
paths in $\automaton$ by the reduced sequence of pairs of both their
ends.  Intuitively, the combination of frontiers with $\star$
corresponds to paths concatenation. Note however that $\star$ is not a
function. For instance, with \(f = (x_1, y_1), (x_2,y_1), (x_3, y_1)\)
and \(g = (y_1, z_1), (y_1, z_2)\) one has
$f\star g \to (x_1,z_1), (x_2,z_1), (x_3,z_2)$ and
$f\star g \to (x_1,z_1), (x_2,z_2), (x_3,z_2)$. One way turn it into a
function, is to lift the operator $\star$ to sets of frontiers: for
$F, G \subseteq \calF$,
$F \star G = \{h \mid \exists f\in F,\ g\in G:\ f\star g \to h\}$.  As
we will see, this operation yields a semigroup over $2^{\calF}$.

We start with a technical lemma expressing that if two sequences of
tuples $\bx$ and $\by$ are abtracted by frontiers $f$ and $g$, and
$f \star g \to h$, then there is a sequence of tuples that unifies
$\bx$ and $\by$ and whose abstraction is $h$. Its proof is given in
Appendix~\ref{app:unify-product}. Recall that for a
sequence ${\tau}$, 
${\tau}[i,j]$ denotes the
  sequence composed of the two elements ${\tau}[i],{\tau}[j]$.

\begin{restatable}{lemma}{LemUnifyProduct}
	\label{lem:unify-product} 
	Let $f,g,h \in \calF$ be frontiers such that $f \star g \to
        h$.  Suppose we have sequences
        \( \mathbf{x} = \seq{\tau_1^{\bx} \\ \vdots \\ \tau_r^{\bx}}
        \in \left(Q^m\right)^r\) and
        $\mathbf{y} = \seq{\tau_1^{\by} \\ \vdots \\ \tau_s^{\by}} \in
        \left(Q^n\right)^s$ (for some $m,n,r,s$) such that:
        \[f = \reduction{\seq{\tau_1^{\bx}[1,m] \\
              \vdots \\ \tau_r^{\bx}[1,m]}}
          \qquad
          \text{and}\qquad
          g = \reduction{\seq{\tau_1^{\by}[1,n] \\
              \vdots \\ \tau_s^{\by}[1,n]}}\enspace. \]
        Then there exists
        \(\bz = \seq{\tau_1^{\bz} \\ \vdots \\ \tau_p^{\bz}} \in
        \left(Q ^{m+n-1}\right)^p\) (for some $p$) such that:
        \[
          \reduction{\mathbf{x}} = \reduction{\seq{\tau_1^{\bz}[:m] \\
              \vdots \\ \tau_p^{\bz}[:m]}} \qquad\qquad
          \reduction{\mathbf{y}} = \reduction{\seq{\tau_1^{\bz}[m:] \\
              \vdots \\ \tau_p^{\bz}[m:]}} \qquad\text{and}\qquad h =
          \reduction{\seq{\tau_1^{\bz}[1,m+n-1] \\ \vdots
              \\ \tau_p^{\bz}[1,m+n-1]}}
        \]
\end{restatable}


		
		
		
	

We can now show that the operation $\star$ is associative on
$2^{\calF}$, and thus that $(2^\calF, \star)$ is a semigroup (it is
even a monoid). The proof is given in Appendix~\ref{app:semigroup}.

\begin{restatable}{lemma}{SemiGroup}\label{lem:semigroup}
$(2^\calF, \star)$ is a semigroup.
\end{restatable}

As the $\star$ operation is associative, the notation
$F_1 \star \cdots \star F_k$ is well-defined. An element of this
product can be decomposed into simpler elements, as follows. The proof
is given in Appendix~\ref{app:decompose}.

\begin{restatable}{lemma}{DecomposeFrontier}
	\label{lem:decompose-frontiers}
	For all $F_1, \ldots, F_n \in 2^\calF$ and $f \in \calF$, we
        have $h \in F_1\star \cdots \star F_n$ if and only if there
        exist $f_1 \in F_1, \ldots, f_n \in F_n$ and a sequence of
        tuples
        $\bfx = \seq{\tau_1^{\bx} \\ \vdots \\ \tau_p^{\bx}} \in \left(Q^{n+1}\right)^p$
        such that
        \[
          h = \reduction{\seq{\tau^{\bx}_1[1,n+1] \\ \vdots \\
              \tau^{\bx}_p[1,n+1]}} \qquad \text{and for all
            $1 \le i \le n$, } \quad f_i =
          \reduction{\seq{\tau^{\bx}_1[i:i+1] \\ \vdots \\
              \tau^{\bx}_p[i:i+1]}} .
        \]
\end{restatable}

We call \((2^\F, \star)\) the \emph{frontier semigroup}.  To relate
(winning) strategies and frontiers, we use wins words, and we consider
a morphism \(\psi : \{0, 1\}^+ \to 2^\F\) whose aim is to exhibit
winning strategies.  Informally, \(\psi(0)\) describes sequences of
transitions in which a $\good$ label appears, with no $\bad$ before,
and \(\psi(1)\) describes frontiers corresponding to any sequence of
transitions.

\begin{definition}
	\label{def:morphism-arena}
	Define the morphism \(\psi: \set{0,1}^* \to 2^\calF\) as
	\[
		\psi(0) = \Bigg\{
		\seq{
			x_1, y_1 \\
			\vdots \\
			x_p, y_p}
		\in \mathcal{F} \mid 
		\exists \seq{a_1\\ \vdots\\ a_p},  \forall i, \delta(x_i, a_i) = y_i 
		 \land  \exists j, \lambda(x_j, a_j) = \good
		\land \forall i < j, \lambda(x_i, a_i) = \neutral
		\Bigg\}
	\]
      \[ \psi(1) = \left\{\seq{x_1, y_1 \\ \vdots \\ x_p, y_p} \in \F \mid
        \exists \seq{a_1\\ \vdots \\ a_p}, \forall i, \delta(x_i, a_i)=y_i\right\}\]
\end{definition}

	

\begin{example}
  Consider the instance described in Figure~\ref{fig:Ex1}, and the
  grid from Figure~\ref{fig:possible-visited-statesbis}. The set
  \(\psi(00)\) includes the compressed version of columns 1 and 2 on
  the one hand, and 2 and 3 on the other hand. With our notations,
  this translates into:
  \[
    \seq{q_0, q_1 \\
      q_0, q_0}, \seq{
      q_1, q_1 \\
      q_0, q_1 \\
      q_0, q_0} \in \psi(00).
  \]
\end{example}

Lemma \ref{lem:morphism-describes-slices-of-winning-paths} (whose
proof is given in Appendix~\ref{app:morphism}) formalises the
application of morphism $\psi$ to an arbitrary word. %
The idea is that pairs of a frontier are played step by step, and
whenever the first \(p\) pairs have already been played, the playable
pairs are all of the first \(p+1\) pairs.

\begin{restatable}[Morphism $\psi$ describes slices of winning paths]{lemma}{LemMorphismSlices}
	\label{lem:morphism-describes-slices-of-winning-paths}
	For all \(w \in \{0, 1\}^n\) and \(h \in \F\), we have
        \(h \in \psi(w)\) if and only if there exist words
        \(v_1, \ldots, v_p \in \Sigma^n\) and
        \(q_1, \ldots, q_p \in Q\) such that
        \[
          w \overset{q_1 \cdot v_1}{\rightsquigarrow} \ldots
          \overset{q_p \cdot v_p}{\rightsquigarrow} 1^n \qquad
          \text{and}\qquad
          h = \reduction{\seq{q_1, \delta(q_1, v_1) \\
              \vdots \\ q_p, \delta(q_p, v_p)}}.
        \]
        \end{restatable}

\subsection{Using frontiers to decide \reachtogether}

Now that we have defined a finite semigroup to abstract finite
intervals of columns of population strategies, we use it to obtain
witnesses of existence of a winning strategy. We show that winning
strategies can be summarised by two frontiers, one representing the
first few columns, and the other representing a somewhat periodic
pattern in the remaining columns.

\begin{definition}
	A frontier is \emph{initial} if all its pairs are in \(\{q_{\text{init}}\} \times Q\).
	A frontier \(g\) is \emph{\(\omega\)-iterable} if there exist
        \(x_1, \ldots, x_p, y_1, \ldots, y_r, z_1, \ldots, z_r\in Q\) such that
	\( g = \seq{x_1, x_1 \\ \vdots \\ x_p, x_p \\ y_1, z_1 \\
            \vdots \\ y_r, z_r} \)
          with 
          \(z_i \in \{x_1, \ldots, x_p, y_1, \ldots, y_{i-1}\}\) for all \(i\).
\end{definition}

When a sequence of paths in the automaton is sliced into intervals of
columns, themselves abstracted into frontiers, the first one is
initial.  An \(\omega\)-iterable frontier starts with
\((x_i, x_i)\)-pairs that can be all concretized
by one infinite move. The next ones have to be concretized
by iterating families of infinite moves, covering the slices associated with the
\(\omega\)-iterable frontier in increasing order. 

\begin{example}
  Consider the automaton in
  Figure~\ref{fig:positive-regular-instance-with-node-label}, and a
  winning strategy described in
  Figure~\ref{fig:possible-visited-states-regular}. Consider frontiers
  $f$ and $g$ as described on  Figure~\ref{fig:possible-visited-states-regular}.
  The first frontier
  is \(f = \seq{q_0, q_2 \\ q_0, q_4 \\ q_0, q_3}\) and is initial,
  the second one is
  \(g = \seq{q_2, q_2 \\ q_4, q_4 \\ q_3, q_4 \\ q_3, q_3}\) and is
  \(\omega\)-iterable. 
  
  \begin{figure}[H]
        \begin{subfigure}{0.35\textwidth}
      \begin{tikzpicture}[AUT style]
		\node[state, initial]  (0) {};
		\node[state, above right=of 0]  (1) {};
		\node[state, right=of 1]  (2) {};
		\node[state, below right = of 0]  (3)  {};
		\node[state, right=of 3]  (4) {};

		\node[below=of 0, yshift=1.2cm]  (0a) {\(q_0\)};
		\node[below=of 1, yshift=1.2cm]  (1a) {\(q_1\)};
		\node[below=of 2, yshift=1.2cm]  (2a) {\(q_2\)};
		\node[below=of 3, yshift=1.2cm]  (3a) {\(q_3\)};
		\node[below=of 4, yshift=1.2cm]  (4a) {\(q_4\)};
		
		\draw[->, color=Green2] (0) edge node[above left] {$\good$} (1);
		\draw[->,color=Blue2](0) edge node[below left] {$\neutral$} (3);
		\draw[->,in=75, out=105, loop, color=Red2] (3) edge node[above,yshift=-2pt] {$\bad$} ();
		\draw[->,in=75, out=105, loop, color=Blue2] (4) edge node[above,yshift=-3pt] {$\neutral$} ();
		\draw[->, color=Green2](3) edge node[above] {$\good$} (4);
		\draw[->,bend left=45, looseness=0.75,color=Blue2] (1) edge node[above,yshift=-2pt] {$\neutral$} (2);
		\draw[->, bend left=45, looseness=0.75,color=Green2]
                (2) edge  node[below] {$\good$} (1);
                
                \node[below=of B, yshift=-10mm] (C) {}; 
\end{tikzpicture}
      \caption{The labelled automaton.}
      \label{fig:positive-regular-instance-with-node-label}
    \end{subfigure}
 \hfill
    \begin{subfigure}{0.6\textwidth}
\begin{tikzpicture}
	
	\node (04) at (0,2.4) {$q_0$};
	
	\foreach \x in {1,...,6}
	{
		\pgfmathparse{mod(\x,2) ==0 ? 1 : 0}
		\ifthenelse {\pgfmathresult>0}{
			\pgfmathparse{int(\x-1)}
			\node (\x4) at (1.2*\x,2.4) {$q_2$};
			\draw[->, color=Blue2] (\pgfmathresult4) edge node[above, yshift=-3pt] {\tiny$\neutral$} (\x4);	
	}{
		\pgfmathparse{int(\x-1)}
		\node (\x4) at (1.2*\x,2.4) {$q_1$};
		\draw[->, color=Green2] (\pgfmathresult4) edge node[above, yshift=-3pt] {\tiny$\good$} (\x4);
	};
	}

	\foreach \x in {1,...,6}
	{
		\foreach \y  in {0,...,3}
		{
			\pgfmathparse{\x+2*\y-8 < 0 ? 1 : 0}
			\ifthenelse {\pgfmathresult>0}{
				\node (\x\y) at (1.2*\x,0.6*\y) {$q_3$};
				\ifthenelse{\x>1}
				{
					\pgfmathparse{int(\x-1)}
					\draw[->, color=Red2] (\pgfmathresult\y) edge node[above, yshift=-3pt] {\tiny$\bad$} (\x\y);	
				}{};
			}{};
			
			\pgfmathparse{\x+2*\y-8 == 0 ? 1 : 0}
			\ifthenelse {\pgfmathresult>0}{
				\node (\x\y) at (1.2*\x,0.6*\y) {$q_4$};
				\pgfmathparse{int(\x-1)}
				\draw[->, color=Green2] (\pgfmathresult\y)  edge node[above, yshift=-3pt] {\tiny$\good$} (\x\y);
			}{};	
			
			\pgfmathparse{\x+2*\y-8 > 0 ? 1 : 0}
			\ifthenelse {\pgfmathresult>0}{
				\node (\x\y) at (1.2*\x,0.6*\y) {$q_4$};
				\pgfmathparse{int(\x-1)}
				\draw[->, color=Blue2] (\pgfmathresult\y) edge node[above, yshift=-3pt] {\tiny$\neutral$} (\x\y);
			}{};
		}
	}

	\foreach \y  in {0,...,3}
	{
		\node (0\y) at (0,0.6*\y) {$q_0$};
		\draw[->, color=Blue2] (0\y) edge node[above, yshift=-3pt] {\tiny$\neutral$} (1\y);	
	}
	
	\node[below= of 32, yshift=-3mm, xshift=-3mm] (vd) {\LARGE $\vdots$};
	
	\foreach \y  in {0,...,4}
	{
		\node (dots\y) at (8, 0.6*\y) {\Large $\cdots$};
	}

	\draw[bend left,densely dotted] (04.north west) edge node[above] {$f$} (24.north east);
           \draw[densely dotted] (00.south west) -- (04.north west);
    \draw[densely dotted] (20.south east) -- (24.north east);
    \draw[bend left,dashed] (24.north west) edge node[above] {$g$} (44.north east);
    \draw[dashed] (20.south west) -- (24.north west);
    \draw[dashed] (40.south east) -- (44.north east);
    \draw[bend left,dashed] (44.north west) edge node[above] {$g$} (64.north east);
    \draw[dashed] (40.south west) -- (44.north west);

\end{tikzpicture}
\caption{A possible winning population strategy.}
      \label{fig:possible-visited-states-regular}
    \end{subfigure}
    \caption{A positive instance of \reachtogether, and frontiers on its winning strategy.}
  \end{figure}
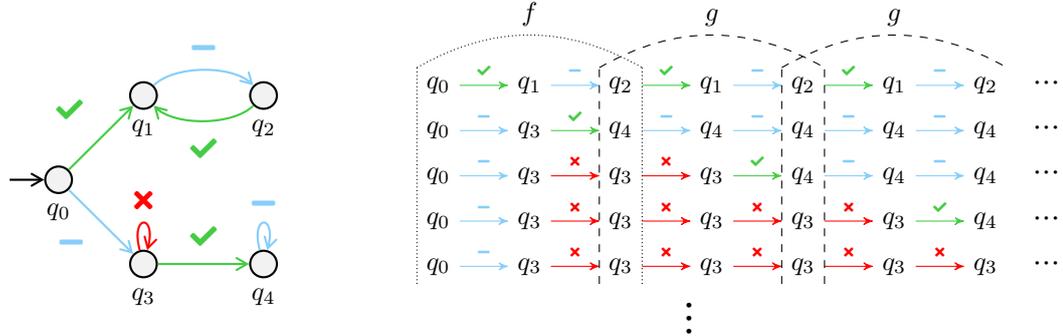
\end{example}



We start with the observation that if two frontiers can be composed,
then the order of appearance of states on their ``interface'' must be the
same (proof in Appendix~\ref{app:order}).

\begin{restatable}{lemma}{LemSameOrder}
	\label{lem:f-g-same-order}
	Let $f, g \in \calF$ with $f = \seq{x_1, y_1 \\ \vdots \\ x_r,
          y_r}$ and
        $g = \seq{x'_1, y'_1 \\ \vdots \\ x'_s, y'_s}$.  Assume
        \(f * g \to h\) for some $h \in \calF$. Then we have
        $\reduction{\seq{y_1 \\ \vdots \\ y_r}} = \reduction{\seq{x'_1
            \\ \vdots \\ x'_s}}$.
\end{restatable}

We can now get to the core of the proof: We will first show how to
construct a winning strategy from an initial and an $\omega$-iterable
frontier (Theorem~\ref{th:omega-iterable-can-iterate}).  Then we will
show how to extract such frontiers from an arbitrary winning strategy
(Theorem~\ref{th:strategy-frontiers}).  The two results will give us
the desired characterisation of automata for which there exist winning
population strategies
(Theorem~\ref{th:characterization-positive-instances-with-frontiers}).


A \emph{slice} of length $n \in \NN_{>0} \cup \{\omega\}$ is a finite
sequence of pairs
$(q_i, v_i)_{1 \leq i \leq p} \in \left(Q \times
  \Sigma^n\right)^p$. Note that while a slice is always a finite
sequence, its length can be infinite, i.e, it can contain infinite
words.  Given a sequence of distinct states $x_1, \ldots, x_k$, we say
that the slice starts from $x_1, \ldots, x_k$ if
$\reduction{q_1, \ldots, q_p} = x_1, \ldots, x_k$ and that it ends in
$x_1, \ldots, x_k$ if it has finite length and
$\reduction{\delta(q_1, v_1), \ldots, \delta(q_p,v_p)} = x_1, \ldots,
x_k$.
The \emph{result} of the slice is the wins word $w \in \{0,1\}^n$
such that
$0^n \overset{q_1 \cdot v_1}{\rightsquigarrow} \cdots \overset{q_p
  \cdot v_p}{\rightsquigarrow} w$, if defined.

\newcommand{\slice}{\mathbf{s}}
\newcommand{\bs}{\slice}

We will construct a winning strategy by constructing slices from frontiers and then patching them together.
The patching operation is described in the following lemma, whose
proof is in Appendix~\ref{app:compose}.

\begin{restatable}{lemma}{LemComposeSlices}
  \label{lem:compose-slices}
  Suppose we have a slice $\slice$ starting from
  $x_1, \ldots, x_{p(x)}$ and ending in $y_1, \ldots, y_{p(y)}$ with
  result $w$ and another one $\slice'$ starting from
  $y_1, \ldots, y_{p(y)}$ with result $w'$.  Then we have a slice
  $\slice''$ starting from $x_1, \ldots, x_{p(x)}$ with result $w w'$.
	
  Furthermore if $\slice'$ is finite and ends in
  $z_1, \ldots, z_{p(z)}$, then so does $\slice''$.
  \end{restatable}

This allows to state the first implication of the main result. We only
sketch the proof here, and defer the full proof to Appendix~\ref{app:iterable}.
  
\begin{restatable}[Frontiers to strategy]{theorem}{ThmFrontStrat}
  \label{th:omega-iterable-can-iterate}
  Assume there exist $f,g \in \psi(0^+)$ such that:
  \begin{itemize}
    \begin{minipage}[t]{.3\linewidth}\item \(f\) is initial,
    \item \(g\) is \(\omega\)-iterable,
    \end{minipage}
    \begin{minipage}[t]{.7\linewidth}
    \item \(f * g \to f\),
    \item \(g * g \to g\).
    \end{minipage}
  \end{itemize}\smallskip
  Then there exists a winning population strategy in \(\langle \calA,\lambda \rangle\).
\end{restatable}

\begin{proof}[Sketch of proof]
  We translate the frontiers $f$ and $g$ into slices using
  Lemma~\ref{lem:morphism-describes-slices-of-winning-paths}. We show
  that we can win on every position using the following process: we
  consider chunks of positions from left to right one by one. On each
  one, we apply the slice given by $g$ to win on those positions. The
  initial frontier $f$ lets us complete this slice into a finite
  sequence of moves: roughly speaking, we use the slice given by $f$
  to extend it to the left and infinitely many iterations of prefixes
  of the slice of $g$ to extend it to the right (which is possible
  thanks to $\omega$-iterability).  We patch those slices together
  using Lemma~\ref{lem:compose-slices}.
\end{proof}

We now need to show the other implication, i.e., that a winning strategy yields suitable frontiers. The proof will use Ramsey's theorem on
infinite graphs, which we recall below.

\begin{theorem}[Ramsey's theorem on infinite graphs~\cite{Ramsey}]
  \label{th:ramsey}
  In an infinite complete graph with edges coloured by finitely many
  colours, there is an infinite clique whose edges all have the same
  colour.
\end{theorem}

This use of Ramsey's theorem is common in the study of the interplay
between finite semigroups and infinite words: if we have an infinite
sequence of elements of a finite semigroup, we can cut it into a
finite prefix and infinitely many factors that all evaluate to the
same (idempotent) element in the semigroup.  This idea appears
already, for instance, in the complementation construction for Büchi
automata by the eponymous author~\cite{Buchi60}.

\begin{restatable}[Strategy to frontiers]{theorem}{ThmStratFront}
  \label{th:strategy-frontiers}
  If there is a winning population strategy in \(\calG(\A)\), then
  there exist 
  \(f,g \in \psi(0^+)\)
  such that
  \begin{itemize}
    \begin{minipage}[t]{.3\linewidth}\item \(f\) is initial,
    \item \(g\) is \(\omega\)-iterable,
    \end{minipage}
    \begin{minipage}[t]{.7\linewidth}
    \item \(f * g \to f\),
    \item \(g * g \to g\).
    \end{minipage}
  \end{itemize}\smallskip
\end{restatable}

\begin{proof}[Sketch of proof]
  This proof is made of different steps :
  \begin{enumerate}
  \item We build an infinite graph with coloured edges, which depends
    on the winning strategy;
  \item We use Ramsey's Theorem on this graph to extract $f$ and \(g\).
  \end{enumerate}
  
  Let \(\sigma\) be a winning strategy, and consider the
  representation where moves are depicted in successive lines as in
  Figure~\ref{fig:possible-visited-states-regular}. We build a graph
  whose set of vertices is \(\NN\), and whose edge between \(k\) and
  \(\ell\) is coloured by a pair $(f,g)$, where:
  \begin{itemize}
  \item $f$ is the frontier summarizing the strategy between the
    initial column and column $k$,
  \item $g$ is the frontier summarizing the strategy between column
    $k$ and column $\ell$.
  \end{itemize}
  As the labelled automaton has a finite number of states, the
  graph has a finite number of colours, and we can apply Ramsey's
  theorem on it.

  Let \((f,g)\) be the colour of the infinite subgraph given by
  Ramsey. We note that $f$ is an initial frontier, since all moves
  given by $\sigma$ start at $q_{\text{init}}$. By definition of the
  colour of an edge, the right part of $f$ coincides with the left
  part of $g$ as well as the right part of $g$. This allows to infer
  that \(f * g \to f\) and \(g * g \to g\). The $\omega$-iterability
  of $g$ follows from the fact that \(g\) is indeed iterated
  infinitely many times by $\sigma$.
\end{proof}

The full proof is given in Appendix~\ref{app:frontiers}.

\section{The complexity of \reachtogether}
\label{sec:complexity}
In this section, we establish the precise complexity of
\reachtogether.  While
Theorem~\ref{th:characterization-positive-instances-with-frontiers}
provides witnesses for the existence of winning strategies, we now
show how to look for those witnesses in polynomial space. After, we
exhibit a matching complexity lower bound.

\begin{theorem}
\reachtogether is \PSPACE-complete.
\end{theorem}

\subsection{Complexity upper bound}
To begin with, we observe that deciding if a frontier is in $\psi(0^+)$ comes down to exploring a graph of exponential size.

\begin{lemma}
	\label{lem:PSPACEimagePsi}
	Given a frontier $f \in \calF$, one can check in polynomial space whether $f \in \psi(0^+)$.
\end{lemma}

\begin{proof}
	We build the following graph, of exponential size.
	Vertices are frontiers, and edges go from each frontier to the ones that can be obtained from it by composition with an element of $\psi(0)$.
	Formally, $G = (\calF, E)$, with $(g,h) \in E$ if and only if there exists $g' \in \psi(0)$ such that $g \star g' \to h$.
	This condition is easily checked in polynomial space. 
	
	Clearly $f \in \psi(0^+)$ if and only if there is a path in $G$ from an element of $\psi(0)$ to $f$. This can be checked by guessing a path in $G$ on the fly. Since one can store a vertex and check the existence of an edge in polynomial space, we get a non-deterministic polynomial space algorithm. We conclude using Savitch's theorem.
\end{proof}

\begin{proposition}
  \label{prop:inPSPACE}
  \reachtogether is in \PSPACE.
\end{proposition}

\begin{proof}
  We non-deterministically guess $f, g \in \calF$ so that $f$ is initial, $g$ is $\omega$-iterable, $f \star g \to f$ and $g \star g \to g$.
	Those conditions can easily be checked in \PSPACE.
	
	Then, we check that there exist $k,l \in \NN_{>0}$ such that $f \in \psi(0^k)$ and $g \in \psi(0^l)$, i.e., if $f,g \in \psi(0^+)$. This condition can also be checked in polynomial space by Lemma~\ref{lem:PSPACEimagePsi}.
\end{proof}

\subsection{Complexity lower bound}

We match the \PSPACE upper bound with a lower bound, reducing from the
termination problem for deterministic Turing machines (DTMs) with unary
bounded space.  We define DTMs as usual, and simply fix the notations.

We fix a finite alphabet $\Sigma$ and define a deterministic Turing machine  as a tuple $\mathcal{M} = (S_\calM, \Sigma, \delta_\calM,s_{\text{init}}, F)$, with $\delta_\calM : S_\calM \times \Sigma \to S_\calM \times \Sigma \times \set{\leftarrow, \rightarrow}$ the transition, reading a state and a letter, updating both and then moving left or right on the tape. 

A configuration of an $n$-bounded Turing machine is a word of the form $u (s,a) v$ with $u,v \in \Sigma^*$, $s \in S_\calM$ and $a \in \Sigma$, so that $|u| +|v| +1 =n$.
The initial configuration is $(s_{\text{init}},b)b^{n-1}$.

The input is a deterministic Turing machine along with a number $n \in \NN$ in unary.
The question is whether the run from the initial configuration eventually reaches a configuration with a state in $F$.
This problem is well-known to be \PSPACE-complete.

The intuition of the reduction is the following. We see configurations
of the Turing machine as words of fixed length over an alphabet
$\Gamma = \Sigma \times (\Sigma \times S_\calM) \cup \set{\#}$, with
$\#$ a fresh letter. We then in turn encode those letters as words of
the form $0^i 1 0^{|\Gamma|-i-1}$ over $\set{0,1}$.  Note that all
those words have the same length $|\Gamma|$.  The sequence of
configurations of the run of the DTM can then be seen as a
(potentially infinite) word over $\set{0,1}$, two consecutive
configurations being separated by a $\#$.  Since the machine is
deterministic, for all $i \geq n +1$, the $i$-th letter (as a word
over $\Gamma$) is determined by the $(i-n-1)$-th, $(i-n)$-th and
$(i-n+1)$-th letters.

We build an automaton that computes this sequence by \emph{reading} and \emph{writing}.
Reading a letter $0^i 1 0^{|\Gamma|-i-1}$ means going through a sequence of states with labels $\neutral^i \bad \neutral^{|\Gamma|-i-1}$. This will lead to a loss if the sequence of bits on these positions is of the form $0^j 1 0^{|\Gamma|-j-1}$ with $j \neq i$.
Writing $0^i 1 0^{|\Gamma|-i-1}$ means going through a sequence of states with labels $\neutral^i \good \neutral^{|\Gamma|-i-1}$. If the sequence of bits was $0^{|\Gamma|}$ before, we will obtain the desired sequence. If it was already $0^i 1 0^{|\Gamma|-i-1}$, nothing changes.
This is where the determinism of the machine is important: the automaton will never be able to write two different letters at the same position.
Our automaton will have a path writing the initial configuration. It will also be able to read three consecutive letters at any positions $i-1,i, i+1$, infer the $(i+n+1)$-st letter, skip $n+1$ positions and write that letter.
Finally, it can read a letter from $\Sigma \times F$, and turns all following positions to $1$, as well as the position just before.
Hence if we encounter a final state during the computation, this branch will allow us to turn all positions to $1$ by applying it sufficiently many times.

The detailed reduction and the proof of its correctness can be found
in Appendix~\ref{App:LowerBound}.

\section{Conclusion and discussion}
\label{sec:conclusion}

We have shown \PSPACE-completeness of a new class of repeated games,
tailored to population models.  The problem reduces to a restricted,
but challenging, form of grid tiling problem, which we solve using new
algebraic techniques. This result lets us hope to push the
decidability frontier further: In particular, we may be able to expand
on the techniques developed here to solve the open problem of
parameterized concurrent reachability games on arbitrary finite
arenas, mentioned as open in~\cite{BBM-fsttcs20}. Since \reachtogether
can be formulated as a grid tiling problem, our result might also be
useful to better understand some open tiling problems such as those
mentioned~\cite{BlumensathCC14}.

Beyond deciding the existence of a winning population strategy
\emph{for all population sizes}, determining assumptions on the
population sizes that guarantee the existence of a winning strategy
seems hard. Let us elaborate on that natural extension of
\reachtogether.  Given a labelled automaton
$\langle \calA,\lambda \rangle$ for the utility functions, one can
define its \emph{winning region} as the set of wins words
$w \in \set{0,1}^*$ from which there is a winning population strategy,
i.e., such that there is a sequence
$w_0 \rightsquigarrow w_1 \rightsquigarrow \cdots$ with $w_0 = w$ and
such that for all $i$ there exists $j$ such that $w_j[i] = 0$.
\reachtogether corresponds to deciding whether $0^\omega$ belongs to
the winning region. In light of
Lemma~\ref{lem:morphism-describes-slices-of-winning-paths}, one can
wonder if it is possible to extend our construction and \PSPACE
algorithm to compute the winning region. Indeed, morphism $\psi$
applies to all words in $\set{0,1}^*$.  We provide here a partial,
somewhat surprising, answer: in general, the winning region is not an
$\omega$-regular language (the reader is referred to~\cite{Thomas90}
for definitions and properties of $\omega$-regular languages), as
illustrated by the following example:
\begin{example}
  Let $W \subseteq \set{0,1}^\omega$ be the winning region for the
  population game associated with the automaton from
  Figure~\ref{fig:not-w-reg}.
  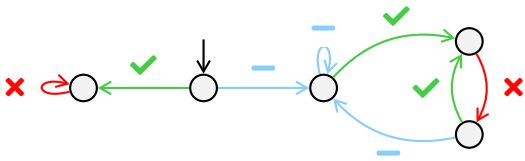
\begin{figure}[h]
		\begin{tikzpicture}[AUT style]
	
	\node[state, initial, initial where=above] (A) {};
	\node[state ,right= of A] (B1)  {};
	\node[state,above right= of B1, yshift=-5mm, xshift=8mm] (B2) {};
	\node[state,below right = of B1,  yshift=5mm, xshift=8mm] (B3)  {};
	
	\draw[->, bend left, color=Green2] (B1) edge[above] node[xshift=3pt,yshift=-1.25pt] {$\good$} (B2);
	\draw[->, bend left, color=Red2] (B2) edge[right] node {$\bad$} (B3);
	\draw[->, bend left, color=Blue2] (B3) edge[below] node[yshift=1pt] {$\neutral$} (B1);
	\draw[->, bend left, color=Green2] (B3) edge[left] node[xshift=2pt] {$\good$} (B2);

	\draw[->, color=Blue2] (A) edge[above] node[yshift=-3pt] {$\neutral$} (B1);
	\draw[->, loop above, color=Blue2] (B1) edge[above] node[yshift=-3.5pt] {$\neutral$} ();

	\node[state,left = of A] (C1)  {};
	
	\draw[->, color=Green2] (A) edge[above] node[yshift=-2pt] {$\good$} (C1);
	\draw[->, loop left, color=Red2] (C1) edge[left] node {$\bad$} ();
\end{tikzpicture}
		\caption{An automaton whose associated winning region is not $\omega$-regular.}
		\label{fig:not-w-reg}
	\end{figure}
	There are two ways to get new 1s in a word. We can either use
        the left branch, which turns the first position to $1$ but
        needs all following positions to be $1$ already.  Or we can
        use the right branch. It lets us put a $1$ on all positions
        which are followed by a $1$. In other words, we can turn the
        last $0$s of all blocks of $0$ to $1$s.  Suppose we start with
        a configuration in $0(10^+)^\omega$. Then there is a winning
        strategy if and only if we can eliminate all $0$s (except the
        first one) by applying the right branch finitely many times.
        This is the case if and only if the blocks of $0$ are
        uniformly bounded.
	
	We obtain that $W \cap (10(01)^+)^\omega$ is the set of words
        with uniformly bounded blocks of $0$, which is not an
        $\omega$-regular language.
\end{example}

\bibliography{biblio}

 \appendix
 \newpage
 \pagebreak

 \section{Proof of Lemma~\ref{lem:unify-product}}
\label{app:unify-product}

\LemUnifyProduct*

\begin{proof}
	Let $f,g,h$ and $\mathbf{x}, \mathbf{y}$ as in the lemma statement.
	Since $f \star g \to h$, there exist
        $\mathbf{w} = \seq{\tau^\bw_1 \\ \vdots \\ \tau^\bw_d} \in
        \left(Q^3\right)^d$ such that
        \[
          f = \reduction{\seq{\tau^\bw_1[:2] \\ \vdots \\
              \tau^\bw_d[:2]}} \qquad\qquad g =
          \reduction{\seq{\tau^\bw_1[2:] \\ \vdots \\ \tau^\bw_d[2:]}}
          \qquad \text{and} \qquad h =\reduction{\seq{\tau^\bw_1[1,3] \\ \vdots \\
              \tau^\bw_d[1,3]}}.
        \]
        The table $\bw$ explains somehow how to glue together $f$ and
        $g$ to get $h$.
        
        Given indices $i,j,k$, we say that $\tau_i^{\bx}$,
        $\tau_j^{\by}$ and $\tau_k^\bw$ are \emph{compatible} if
        $\tau_i^\bx[1] = \tau^\bw_k[1]$,
        $\tau_i^\bx[m] = \tau_j^\by[1] = \tau^\bw_k[2]$ and
        $\tau_j^\by[n] = \tau_k^\bw[3]$. If they are compatible, we
        write $\tau_i^\bx \cdot \tau_j^\by$ for the sequence
        $\tau^\bx_i[1],\ldots, \tau^\bx_i[m], \tau_j^\by[2], \ldots,
        \tau_j^\by[n]$. 
        Note that $\tau_1^{\bx}$, $\tau_1^{\by}$ and $\tau_1^\bw$ are
        \emph{compatible}.

        We construct $\mathbf{z}$ using the following algorithm.

        \begin{algorithm}[H]
          \caption{}
		\begin{algorithmic}
			\State $\alpha, \beta, \gamma \gets 0$
			\State $\mathbf{z} \gets$ empty list		
			\While{$\alpha < r \lor \beta < s \lor \gamma < p$}
			\State $\alpha', \beta', \gamma' \gets \alpha, \beta, \gamma$
			\ForAll{$i,j,k$ with $i \leq \alpha +1$, $j \leq \beta +1$, $k \leq \gamma +1$}
			\If{$\tau_i^{\bx}$, $\tau_j^{\by}$ and $\tau_k^\bw$ are compatible}
			\State Append $\tau_i^{\bx} \cdot \tau_j^{\by}$ to $\mathbf{z}$
			\State $\alpha' \gets \max(\alpha', i+1)$, $\beta' \gets \max(\beta', j+1)$, $\gamma' \gets \max(\gamma', i+1)$
			\EndIf
			\EndFor
			\State $\alpha \gets \alpha'$, $\beta \gets \beta'$, $\gamma \gets  \gamma'$
			\EndWhile
		\end{algorithmic}
		\label{algo1}
              \end{algorithm}
	
              Let
              $\seq{\tau_1^{\bz} \\ \vdots \\ \tau_t^{\bz}} \in
              \left(Q^{m+n-1}\right)^t$
              be the value of $\mathbf{z}$ after each update according
              to the algorithm.
              We say that a tuple $\tau_i^\bx$ (resp. $\tau_j^\by$,
              $\tau_k^\bw$) \emph{appears in } $\mathbf{z}$ (at step
              $t$) at index $\ell$ if $\tau_\ell^\bz[:m] = \tau_i^\bx$
              (resp.  $\tau_\ell^\bz[m:] = \tau_j^\by$,
              $\tau_\ell^\bz[1,m,m+n-1] = \tau_j^\bw$).  The order of
              appearance of tuples in $\bz$ is the order on the
              indices of their first appearances.

	\begin{claim}
		\label{claim:soundness-tech-lemma}
		The algorithm satisfies the following invariant:

                At step $t$, we have:
		\begin{itemize}
			\item $\tau_1^\bx, \ldots, \tau_\alpha^\bx$
                          all appear in $\mathbf{z}$ up to $t$, in that order.  
			
			\item $\tau_1^\by, \ldots, \tau_\beta^\by$ all
                          appear in $\mathbf{z}$ up to $t$, in that order.
			
			\item $\tau_1^\bw, \ldots, \tau_\gamma^\bw$
                          all appear in $\mathbf{z}$ up to $t$, in that order.
		\end{itemize}
	\end{claim}
	
	\begin{claimproof}
		The invariant clearly holds at the start since  $\alpha, \beta, \gamma$ are all $0$.
		
		Suppose the invariant is satisfied at the start of a
                passage through the while loop.  We prove that it
                holds at the end of that passage.  We only show the
                first item, the two others are proven symmetrically.
		
		We already know that $\tau_1^\bx, \ldots,
                \tau_\alpha^\bx$ all appear in $\bz$, in that order,
                at the start of the passage through the while loop.
		
		In the for loop, if for all $\tau_i^\bx \cdot
                \tau_j^\by$ that we append to the list,
                $i \leq \alpha$ then all those $\tau_i^\bx$ already
                appeared in $\bz$ and the order of appearance is
                unchanged. Since $\alpha$ remains the same, the
                invariant is maintained.  Otherwise, we append some
                $\tau_i^\bx \cdot \tau_j^\by$ with $i = \alpha+1$, and
                we obtain that
                $\tau_1^\bx, \ldots, \tau_{\alpha+1}^\bx$ all appear
                in $\bz$, in that order. As $\alpha'$ is then equal to
                $\alpha+1$, and $\alpha$ is updated to $\alpha'$ at
                the end of the loop, the invariant is maintained.
	\end{claimproof}

        Note that, by design of the algorithm, for every
        $k \le \gamma$, there is $i \le \alpha$ and $j \le \beta$ such
        that $\tau^\bx_i$, $\tau^\by_j$ and $\tau^\bw_k$ are
        compatible.
        
	\begin{claim}
		\label{claim:termination-tech-lemma}
		Algorithm~\ref{algo1} terminates.
	\end{claim}
	
	\begin{claimproof}
          We show that at least one of $\alpha, \beta, \gamma$
          increases at every iteration of the while loop.  Let $\ell$
          be the minimal index such that either
          $\tau^\bw_\ell[:2] = \tau^\bx_{\alpha+1}[1,m]$, or
          $\tau^\bw_\ell[2:] = \tau^\by_{\beta+1}[1,n]$ or
          $\ell=\gamma+1$.

          
          First assume that $\ell \le \gamma$. That means that one of
          the two first cases happen. By symmetry, we assume
          w.l.o.g. that
          $\tau^\bw_\ell[:2] = \tau^\bx_{\alpha+1}[1,m]$.  Applying
          the remark after Claim~\ref{claim:soundness-tech-lemma},
          since $\ell \leq \gamma$, there exists $j \le \beta$ such
          that $\tau_j^\by[1,n] = \tau_\ell^\bw[2:]$.  In particular,
          $\tau_{\alpha+1}^\bx$, $\tau_j^{\by}$ and $\tau^\bw_\ell$
          are compatible, and therefore the algorithm appends
          $\tau_{\alpha+1}^\bx \cdot \tau_j^{\by}$ to the list and
          increases $\alpha'$, and thus $\alpha$.

          Assume that $\ell=\gamma+1$. We first argue that
          $\tau^\bw_\ell[:2] \notin \{\tau^\bx_i[1,m] \mid i \le
          \alpha\}$ implies
          $\tau^\bw_\ell[:2] = \tau^\bx_{\alpha+1}[1,m]$.  This is
          because $\bw$ respects the order given by $\bx$: formally,
          $\reduction{\seq{\tau^\bw_1[:2] \\ \vdots
              \\ \tau^\bw_d[:2]}} = f = \reduction{\seq{\tau_1^{\bx}[1,m] \\
              \vdots \\ \tau_r^{\bx}[1,m]}}$.  Suppose
          $\tau^\bw_\ell[:2] \notin \{\tau^\bx_i[1,m] \mid i \le
          \alpha\}$ and
          $\tau^\bw_\ell[:2] \neq \tau^\bx_{\alpha+1}[1,m]$. Then we
          cannot have $\reduction{\seq{\tau^\bw_1[:2] \\ \vdots
              \\ \tau^\bw_d[:2]}} = \reduction{\seq{\tau_1^{\bx}[1,m] \\
              \vdots \\ \tau_r^{\bx}[1,m]}}$, a contradiction.
          
          Therefore there is $i \le \alpha+1$ such that
          $\tau^\bw_\ell[:2] = \tau^\bx_i[1,m]$. Similarly, there is
          $j \le \beta+1$ such that
          $\tau_\ell^\bw[2:] = \tau^\by_j[1,n]$. In particular,
          $\tau_i^\bx$, $\tau_j^\by$ and $\tau_\ell^\bw$ are
          compatible, hence the algorithm appends
          $\tau_i^\bx \cdot \tau_j^{\by}$ to the list and increases
          $\gamma'$, and thus $\gamma$. Note that it may also increase
          $\alpha'$ or $\beta'$ (hence $\alpha$ or $\beta$), whenever
          $i=\alpha+1$ or $j=\beta+1$.
          %
	\end{claimproof}
          
	As the algorithm terminates, in the end $ \alpha = r$,
        $\beta=s$ and $\gamma=p$. By
        Claim~\ref{claim:soundness-tech-lemma}, the list must
        therefore contain all $\tau_i^{\bx}, \tau_j^{\by}$ and
        $\tau_k^\bw$, and their order of appearance is the same as in
        $\mathbf{x}, \mathbf{y}$ and $\mathbf{w}$.
\end{proof}

\section{Proof of Lemma~\ref{lem:semigroup}}
\label{app:semigroup}

\SemiGroup*

\begin{proof}
	Let $F_1, F_2, F_3 \in 2^\calF$.
	We need to show associativity, i.e., $F_1 \star (F_2 \star F_3) = (F_1 \star F_2) \star F_3$.
	We show the left-to-right inclusion, the other one is proven symmetrically.
	
	Let $h \in F_1 \star (F_2 \star F_3)$, there exist $f_1 \in
	F_1$, $f_2 \in F_2$, $f_3 \in F_3$ and
	$g_{23} \in F_2 \star F_3$ such that
	$f_2 \star f_3 \to g_{23}$ and $f_1 \star g_{23} \to h$.
	
	By definition, there exist $\mathbf{x} = \seq{\tau_1^{\bx} \\
          \vdots \\ \tau_r^{\bx}} \in \left(Q^3\right)^r$ such that
        $f_2 = \reduction{\seq{\tau_1^{\bx}[:2] \\ \vdots \\
            \tau_r^{\bx}[:2]}}$,
        $f_3 = \reduction{\seq{\tau_1^{\bx}[2:] \\ \vdots \\
            \tau_r^{\bx}[2:]}}$ and
        $g_{23} = \reduction{\seq{\tau_1^{\bx}[1,3] \\ \vdots \\
            \tau_r^{\bx}[1,3]}}$.  By Lemma~\ref{lem:unify-product},
        since $f_1 \star g_{23} \to h$, there exists
        $\by = \seq{\tau_1^{\by} \\ \vdots \\ \tau_p^{\by}} \in
        \left(Q^4\right)^p$ such that
        $f_1 = \reduction{\seq{\tau_1^{\by}[:2] \\ \vdots \\
            \tau_p^{\by}[:2]}}$ and
        $\reduction{\mathbf{x}} = \reduction{\seq{\tau_1^{\by}[2:] \\
            \vdots \\ \tau_p^{\by}[2:]}}$.
	
	Observe that the same pairs of states appear in
        $\seq{\tau_1^{\by}[2:3] \\ \vdots \\ \tau_p^{\by}[2:3]}$ and
        $\seq{\tau_1^{\bx}[:2] \\ \vdots \\ \tau_r^{\bx}[:2]}$, in the
        same order, hence
        $f_2= \reduction{\seq{\tau_1^{\by}[2:3] \\ \vdots \\
            \tau_p^{\by}[2:3]}}$.  Similarly,
        $f_3= \reduction{\seq{\tau_1^{\by}[3:] \\ \vdots \\
            \tau_p^{\by}[3:]}}$. Let
        $g_{12} = \reduction{\seq{\tau_1^{\by}[1,3] \\ \vdots \\
            \tau_p^{\by}[1,3]}}$.

        The sequence
        $\seq{\tau_1^{\by}[:3] \\ \vdots \\ \tau_p^{\by}[:3]}$ is a
        witness that $f_1 \star f_2 \to g_{12}$, thus
        $g_{12} \in F_1 \star F_2$.  Similarly, the sequence
        $\seq{\tau_1^{\by} \\ \vdots \\ \tau_p^{\by}}$ witnesses that
        $g_{12} \star f_3 \to h$.
	
	As a result, we have $h \in (F_1 \star F_2) \star F_3$.
\end{proof}

\section{Proof of Lemma~\ref{lem:decompose-frontiers}}
\label{app:decompose}

\DecomposeFrontier*

\begin{proof}
	We proceed by induction on $k$.
	For $k=1$ this is clear, simply take $f_1 = f$.
	
	Let $k>1$, suppose the property holds for $k-1$.
	Let $h \in F_1\star \cdots \star F_k$. There exists $g \in F_1\star \cdots \star F_{k-1}$ and $f_k \in F_k$ such that $g \star f_k \to h$.
	
	By induction hypothesis there exist
        $f_1 \in F_1, \ldots, f_{k-1} \in F_{k-1}$ and a sequence of
        tuples
        $\by = \seq{\tau_1^{\by} \\ \vdots \\ \tau_m^{\by}} \in
        \left(Q^{k}\right)^m$ such that
        \(g = \reduction{\seq{\tau_1^{\by}[1,k] \\ \vdots \\
            \tau_m^{\by}[1,k]}}\) and
	\(f_i = \reduction{\seq{\tau_1^{\by}[i:i+1] \\ \vdots \\
            \tau_m^{\by}[i:i+1]}}\) for all $i<k$.
	
	By Lemma~\ref{lem:unify-product}, since $g \star f_{k} \to h$,
        there exist
        $\seq{\tau_1^{\bx} \\ \vdots \\ \tau_p^{\bx}} \in
        \left(Q^{k+1}\right)^p$ such that:
	\[
          \reduction{\mathbf{y}} = \reduction{\seq{\tau_1^{\bx}[:k] \\
              \vdots \\ \tau_p^{\bx}[:k]}},
          \qquad f_k = \reduction{\seq{\tau_1^{\bx}[k:] \\ \vdots \\
              \tau_p^{\bx}[k:]}},\qquad \text{and}\ h =
          \reduction{\seq{\tau_1^{\bx}[1,k+1] \\ \vdots \\
              \tau_p^{\bx}[1,k+1]}}.
	\]
	
		
		
	
	For all $i<k$, since
	$f_i = \reduction{\seq{\tau_1^{\by}[i:i+1] \\ \vdots \\
            \tau_m^{\by}[i:i+1]}}$ and
        $ \reduction{\seq{\tau_1^{\bx}[:k] \\ \vdots \\
            \tau_p^{\bx}[:k]}} = \reduction{\mathbf{y}} =
	\reduction{\seq{\tau_1^{\by}[:k] \\ \vdots \\
            \tau_m^{\by}[:k]}}$, we can infer that
	$f_i = \reduction{\seq{\tau_1^{\bx}[i,i+1] \\ \vdots \\
              \tau_p^{\bx}[i,i+1]}}$.  This concludes our proof.
        \end{proof}

\section{Proof of Lemma~\ref{lem:morphism-describes-slices-of-winning-paths}}
\label{app:morphism-slices}\label{app:morphism}
\LemMorphismSlices*

\begin{proof}

  We start with the left-to-right implication.  Let
  \(w = b_1 \cdots b_n \in \{0, 1\}^n\) and \(h \in \psi(w)\).  Since
  $\psi$ is a morphism, we have
  $\psi(w) = \psi(b_1) \star \cdots \star \psi(b_n)$.  By
  Lemma~\ref{lem:decompose-frontiers}, there exist
  $f_1 \in \psi(b_1), \ldots, f_n \in \psi(b_n)$ and
  $\bx = \seq{\tau_1^{\bx} \\ \vdots \\ \tau_p^{\bx}} \in
  \left(Q^{n+1}\right)^p$ such that
  $h = \reduction{\seq{\tau_1^{\bx}[1,n+1] \\ \vdots \\
      \tau_p^{\bx}[1,n+1]}}$ and
  $f_i = \reduction{\seq{\tau_1^{\bx}[i:i+1] \\ \vdots \\
      \tau_p^{\bx}[i:i+1]}}$ for all $i$.

  By definition of $\psi$, for all $j$ we have letters
  $a_{j,1}, \ldots, a_{j,n}$ such that
  $\delta(\tau_j^{\bx}[i], a_{j,i}) = \tau_j^{\bx}[i+1]$ for all $i$.
  Furthermore, either $b_i = 1$ or there exists $j$ such that
  $\etiquette{\tau_j^{\bx}[i], a_{j,i}}= \good$ and
  $\etiquette{\tau_{j'}^{\bx}[i], a_{j',i}} = \neutral$ for all
  $j' < j$.

  Define $v_j = a_{j,1} \cdots a_{j,n}$ for all $j$. We obtain that
  \(w \overset{\tau_1^{\bx}[1] \cdot v_1}{\rightsquigarrow} \ldots
  \overset{\tau_p^{\bx}[1] \cdot v_p}{\rightsquigarrow} 1^n
  \). Furthermore,
  $\delta(\tau_j^{\bx}[1], v_{j}) = \tau_j^{\bx}[n+1]$ for all $j$.
  Hence
  $h = \reduction{\seq{\tau_1^{\bx}[1], \delta(\tau_1^{\bx}[1], v_{1}) \\ \vdots \\
      \tau_p^{\bx}[1], \delta(\tau_p^{\bx}[1], v_p)}}$, which is what
  we wanted to show (defining $q_i = \tau_i^{\bx}[1]$).
 
  \smallskip It remains to show the right-to-left implication.  Assume
  there exist words \(v_1, \ldots, v_p \in \Sigma^n\) and
  \(q_1, \ldots, q_p \in Q\) such that
\[
  w \overset{q_1 \cdot v_1}{\rightsquigarrow} \ldots \overset{q_p
    \cdot v_p}{\rightsquigarrow} 1^n \qquad \text{and} \qquad h =
  \reduction{\seq{q_1, \delta(q_1, v_1) \\ \vdots \\ q_p, \delta(q_p,
      v_p)}}.
\]
We use an induction on \(n\).  If \(n=1\), we immediately obtain
\(h \in \psi(w)\) from the definition of $\psi$: if $w = 1$ this holds
trivially. If $w=0$ then
$w \overset{q_1 \cdot v_1}{\rightsquigarrow} \ldots \overset{q_p \cdot
  v_p}{\rightsquigarrow} 1$ implies that there exists $j$ such that
$\etiquette{q_j \cdot v_j} = \good$ and
$\etiquette{q_{j'} \cdot v_{j'}} = \neutral$ for all $j'<j$.

Suppose \(n > 1\) and suppose the property holds for $n-1$.  Assume
\(w=w'b \in \{0,1\}^n\), with \(w' \in \{0,1\}^{n-1}\) and
\(b \in \{0,1\}\).  Each word \(v_i\) can be written \(v'_i a_i\) with
\(v'_i \in \Sigma^{n-1}\) and \(a_i \in \Sigma\). By construction, we
know that
\[
  w' \overset{q_1 \cdot v'_1}{\rightsquigarrow} \ldots \overset{q_p
    \cdot v'_p}{\rightsquigarrow} 1^{n-1}
  \qquad \text{and}\qquad
  b \overset{\delta(q_1, v'_1) \cdot a_1}{\rightsquigarrow} \ldots
  \overset{\delta(q_p, v'_p) \cdot a_p}{\rightsquigarrow} 1.
\]
	
Let us call \(h_{w'}\) and \(h_b\) the following:
\[
  h_{w'} = \reduction{\seq{q_1, \delta(q_1, v'_1) \\ \vdots \\ q_k,
      \delta(q_p, v'_p)}} \qquad \text{and} \qquad h_b =
  \reduction{\seq{\delta(q_1, v'_1), \delta(q_1, v'_1a_1) \\
      \vdots \\ \delta(q_p, v'_p), \delta(q_p, v'_p a_p)}}.
\]
By induction hypothesis, \(h_{w'} \in \psi(w')\) and
\(h_b \in \psi(b)\).  It is easily verified that
$h_{w'} \star h_b \to h$ by applying the definition of $\star$ with
the sequence of triples
$\seq{q_1, \delta(q_1, v'_1), \delta(q_1, v'_1 a_1) \\ \vdots \\ q_p,
  \delta(q_p, v'_p), \delta(q_p, v'_p a_p) }$.  As a consequence,
\(h \in \psi(w'b)\).
\end{proof}

\section{Proof of Lemma~\ref{lem:f-g-same-order}}
\label{app:order}

\LemSameOrder*
\begin{proof}
  Since \(f * g \to h\), there exist
  $(\alpha_1, \beta_1, \gamma_1), \dots, (\alpha_p, \beta_p, \gamma_p)
  \in Q^3$ such that
  $f = \reduction{(\alpha_i, \beta_i)_{1 \leq i \leq p}}$ and
  $g = \reduction{(\beta_i, \gamma_i)_{1 \leq i \leq p}}$.  Hence the
  same states must appear on the first coordinate of $g$ and on the
  second coordinate of $f$, and furthermore they must appear in the
  same order:
  $\reduction{(y_i)_{1 \leq i \leq r}} = \reduction{(\beta_i)_{1 \leq
      i \leq p}} = \reduction{(x'_i)_{1 \leq i \leq s}}$.
\end{proof}

\section{Proof of Lemma~\ref{lem:compose-slices}}
\label{app:compose}

\LemComposeSlices*
\begin{proof}
  Let $\slice = (q_i, v_i)_{1 \leq i \leq p}$ and
  $\slice' = (q'_i, v'_i)_{1 \leq i \leq r}$ as in the statement.
	
  We construct the desired slice $\slice''$ iteratively as follows:
  Start with an empty sequence.  We use two indices $\alpha$ and
  $\beta$, both initialised at $1$.  We proceed as follows: while
  $\alpha \leq p$ or $\beta \leq r$, we have two cases: if there
  exists $i \leq \alpha$ such that $\delta(q_i, v_i) = q'_\beta$, we
  append $(q_i, v_iv'_\beta)$ to $\slice''$ and increment $\beta$.  If
  there exists $i \leq \beta$ such that
  $\delta(q_\alpha, v_\alpha) = q'_i$, we append
  $(q_\alpha, v_\alpha v'_i)$ to $\slice''$.  Since $\slice$ ends in
  $y_1, \ldots, y_{p(y)}$, from which $\slice'$ starts, the order of
  first appearance of the states is the same in
  $\delta(q_1,v_1), \ldots, \delta(q_p,v_p)$ and in
  $q'_1, \ldots, q'_r$. As a consequence, at least one of the two
  cases always holds.
	
  As $\alpha$ or $\beta$ increases at each step, the algorithm
  terminates. Furthermore, it maintains the invariants that $\slice''$
  is a slice starting from $\reduction{q_1, \ldots, q_{\alpha-1}}$,
  either infinite or ending in
  $\reduction{\delta(q'_1, v'_1), \ldots, \delta(q'_{\beta-1},
    v'_{\beta-1})}$.  Further, in light of Lemma~\ref{rem:composition}
  (Composition and repetitions),
  it has result $w_\alpha w'_\beta$, where $w_\alpha$ is the result of
  $(q_i, v_i)_{1 \leq i < \alpha}$ and $w_\beta$ the result of
  $(q'_i, v'_i)_{1 \leq i \leq \beta}$.  In the end, since
  $\alpha = m+1$ and $\beta = n+1$, we obtain the result.
\end{proof}

\section{Proof of Theorem~\ref{th:omega-iterable-can-iterate}}
\label{app:iterable}

\ThmFrontStrat*

\begin{proof}
  Since \(f \in \psi(0^k)\) and \(g \in \psi(0^\ell)\) for some
  $k,\ell \in \NN_{>0}$, by
  Lemma~\ref{lem:morphism-describes-slices-of-winning-paths} we have
  slices
  \((q_i^f,v^f_i)_{1 \leq i \leq m} \in \left(Q\times
    \Sigma^k\right)^m\) and
  \((q_i^g,v^g_i)_{1 \leq i \leq n} \in (Q\times \Sigma^\ell)^n\) with
  result respectively $1^k$ and $1^\ell$.
	
  By Lemma~\ref{lem:f-g-same-order}, since $f \star g \to f$ and
  $g \star g \to g$, the sequences $\reduction{\seq{\delta(q^f_1, v^f_1) \\
      \vdots \\ \delta(q^f_m, v^f_m)}}$, $\reduction{\seq{q^g_1 \\
      \vdots \\ q^g_n}}$ and $\reduction{\seq{\delta(q^g_1, v^g_1) \\
      \vdots \\ \delta(q^g_n, v^g_n)}}$ are equal. Let $\mathbf{q}$ be
  that sequence.
		
  Since $f$ is initial, all $q_i^f$ are $q_{\text{init}}$. As a
  consequence, we have a slice $\slice_{f}$ from
  $\seq{q_{\text{init}}}$ to $\mathbf{q}$.  Since
  \((q_i^g,v^g_i)_{1 \leq i \leq n} \in \left(Q\times
    \Sigma^\ell\right)^n\) is a slice from $\mathbf{q}$ to itself with
  result $1^\ell$, by applying Lemma~\ref{lem:compose-slices} we can
  obtain, for every $N \in \NN_{>0}$, a slice $\slice_N$ from
  $\seq{q_{\text{init}}}$ to $\mathbf{q}$ with result
  $1^{k+\ell\cdot N}$.

  Since $g$ is $\omega$-iterable, it is of the form
  $\seq{x_1, x_1 \\ \vdots \\ x_p, x_p \\ y_1, z_1 \\ \vdots \\ y_r,
    z_r}$ with $z_j \in \{x_1, \ldots, x_p, y_i, \ldots, y_{j-1}\}$
  for all $j$.  For each $j$, let $m(j)$ the minimal index such that
  $q_{m(j)}^g = y_j$. The slice $(q_i^g,v^g_i)_{1 \leq i < m(j)}$
  starts from
  $\reduction{\seq{x_1 \\ \vdots \\ x_p \\ y_1 \\ \vdots \\ y_j}}$ and
  ends in
  $\reduction{\seq{x_1 \\ \vdots \\ x_p \\ y_1 \\ \vdots \\ y_{j'}}}$
  with $j' <j$.  All these slices have well-defined results, as
  prefixes of \((q_i^g,v^g_i)_{1 \leq i \leq n}\), which has a result.

  By Lemma~\ref{lem:compose-slices}, we can thus compose such slices
  to obtain a slice from $\mathbf{q}$ to
  $\mathbf{x} = \seq{x_1 \\ \vdots \\ x_p}$ with some result $w$.  For
  all $N \in \NN$, we can apply the lemma again with $\slice_N$ to get
  a slice $\slice'_N$ from $\seq{q_{\text{init}}}$ to
  $\seq{x_1 \\ \vdots \\ x_p}$ with result $1^{k+\ell \cdot N}w$.
	
  Now consider the minimal index $m_x$ such that
  \( (q_{m_x}, \delta(q_{m_x}, v_{m_x})) = (x_p, x_p)\).  The slice
  \((q_i^g,v^g_i)_{1 \leq i \leq m_x}\) has a result $w_x$, and it can
  be iterated to form an infinite slice: Let
  $\slice_{\omega} = (q_i^g,(v^g_i)^\omega)_{1 \leq i \leq m_x}$. It
  is an infinite slice from $\mathbf{x}$ with result $w_x^\omega$.
  For all $N \in \NN$, we can apply Lemma~\ref{lem:compose-slices}
  with $\slice'_N$ and $\slice_{\omega}$ to obtain an infinite slice
  $(q_{\text{init}}, u_{N,i})_{1 \leq i \leq m_N}$ from
  $\seq{q_{\text{init}}}$ with result
  $1^{k+\ell \cdot N} w w_x^\omega$.
	
  As a consequence, for all $N \in \NN$ we have a sequence of moves
  $0^\omega \overset{u_{N,1}}{\rightsquigarrow} \cdots
  \overset{u_{N,m_N}}{\rightsquigarrow} 1^{k+\ell \cdot N} w
  w_x^\omega$.  We can apply this for every $N \in \NN$ consecutively
  to obtain a winning strategy: by monotonicity (see related item in
  Lemma~\ref{lem:wins-words}), after applying the $N$-th sequence, the
  wins word $w_N$ obtained is well-defined and
  $1^{k+\ell \cdot N} 0^\omega \preceq w_N$.  By
  Lemma~\ref{lem:wins-words} (Winning strategies as wins words), we
  have a winning strategy.
\end{proof}

\section{Proof of Theorem~\ref{th:strategy-frontiers}}
\label{app:frontiers}

\ThmStratFront*
\begin{proof}
	
  Assume there is a population winning strategy
  \(\sigma = (\move_i)_{i \in \NN_{>0}}\) in $\calG(\langle \calA,\lambda \rangle)$. Define
  \(q_{i,j}\) the state of the automaton $\calA$ after reading
  \(\move_i[:j]\):
  \[
    q_{i,j} := \delta(q_{\text{\text{init}}}, \move_i[:j])
  \]
  where $q_{\text{init}}$ is the initial state of $\calA$ (note that
  \(q_{i,0} = q_{\text{init}}\)).  We obtain a grid of states as in
  Figure~\ref{fig:possible-visited-states-regular}.
	
  Consider now the infinite complete undirected graph \(G\) with coloured edges, defined as follows:
  \begin{itemize}
  \item Its vertices are \(\NN\);
  \item Let \(k < \ell \in \NN\), the edge $\set{k,\ell}$ is coloured
    by the pair of frontiers \((f_k,g_{k, \ell})\) where
    \[
      f_k = \reduction{\seq{q_{0,0}, q_{0,k} \\ q_{1,0},q_{1,k} \\
          \vdots}} \qquad \text{and}\qquad g_{k,\ell} =
      \reduction{\seq{q_{0,k},q_{0,\ell} \\ q_{1,k},q_{1,\ell} \\
          \vdots}}.
    \]
    Note that $f_k$ is always initial.
    Informally, when moves are written in lines as in
    Figure~\ref{fig:possible-visited-states-regular},
    \(f_k\) is the frontier obtained from columns between \(0\) and
    \(k\), and \(g_{k, \ell}\) is
    the frontier obtained from columns between \(k\) and \(\ell\).
  \end{itemize}

  We apply Ramsey's theorem (recalled as Theorem~\ref{th:ramsey}) to
  $G$. Let \(S \subseteq \NN\) such that the complete
  sub-graph induced by \(S\) has only one colour \((f, g)\).  By
  construction, \(f\) is initial.  To show that \(g * g \to g\),
  consider \(k < \ell < m \in S\), and the sequence of triples
  \(\seq{q_{0,k},q_{0,\ell},q_{0,m} \\ q_{1,k},q_{1,\ell},q_{1,m} \\
    \vdots }\). By construction:
  \[
    \reduction{\seq{q_{0,k},q_{0,\ell}\\ q_{1,k},q_{1,\ell} \\
        \vdots }} = \reduction{\seq{q_{0,\ell},q_{0,m} \\ q_{1,\ell},q_{1,m} \\
        \vdots }} =\reduction{\seq{q_{0,k},q_{0,m} \\ q_{1,k},q_{1,m} \\
        \vdots }} = g
  \]
  It therefore witnesses the fact that \(g * g \to g\).
  A similar proof shows that \(f * g \to f\), by considering columns
  \(0\), \(i\) and \(j\).
	
  Now we show that \(g \stackrel{\text{def}}{=} \seq{x_1, y_1 \\
    \vdots \\ x_n, y_n}\) is \(\omega\)-iterable.  Let $s : \NN \to S$
  be an enumeration of $S$, i.e., an increasing function such that
  $S = s(\NN)$.  We proceed in two steps:
	
  \begin{itemize}
  \item First we show that for all $i \in \llbracket1,n\rrbracket$,
    $y_i \in \{x_1, \dots, x_{i}\}$.  Let
    $i \in \llbracket1,n\rrbracket$, and let $m$ be the minimal index
    such that there exist $k< \ell \in S$ such that
    $(q_{m,k},q_{m,\ell}) = (x_i,y_i)$.  Such an $m$ exists by
    definition of $g$.  Let $p \in S$ such that $p>\ell$.  By
    minimality of $m$, we must have
    $(q_{j,\ell}, q_{j,p}) \in \{(x_1, y_1),\dots, (x_{i-1},
    y_{i-1})\}$ for all $j<m$.  By definition of $S$, we have
    $\reduction{\seq{q_{0,\ell}, q_{0,p} \\ q_{1,\ell}, q_{1,p} \\
        \vdots}} = g$, thus
    $(q_{m,\ell}, q_{m,p}) \in \{(x_1, y_1),\dots, (x_{i},
    y_{i})\}$. In particular we have $y_i \in \{x_1, \dots, x_{i}\}$.
		
  \item Now we assume by contradiction that $g$ is not
    $\omega$-iterable. Since $y_i \in \{x_1, \dots, x_{i}\}$ for all
    $i$, this means that there exist $j<k$ such that:
    \begin{itemize}
    \item $x_i = y_i$ for all $i<j$,
    \item $x_j \neq y_j$ (hence $y_j \in \{x_1,\dots,x_{j-1}\}$),
    \item $x_k = y_k$ with $y_k \notin \{x_1,\dots,x_{k-1}\}$.
    \end{itemize}

    
    For all $\ell \in \NN$ let $\alpha(\ell)$ be the minimal index
    such that
    $(q_{\alpha(\ell),s(\ell)}, q_{\alpha(\ell),s(\ell+1)}) = (x_j,
    y_j)$. We show that $(\alpha(\ell))_{\ell}$ is an increasing
    sequence. Fix $\ell$:
    $(q_{\alpha(\ell),s(\ell)},q_{\alpha(\ell),s(\ell+1)}) =
    (x_j,y_j)$ with $y_j \in \{x_1,\dots,x_{j-1}\}$, and for every
    $i < \alpha(\ell)$,
    $q_{i,s(\ell)} = q_{i,s(\ell+1)} \in \{x_1,\dots,x_{j-1}\}$. This
    implies that $\alpha(\ell+1)>\alpha(\ell)$.

    By assumption, there is $m$ such that
    $(q_{m,s(0)},q_{m,s(1)}) = (x_k,x_k)$. Since
    $x_k \notin \{x_1,\dots,x_{k-1}\}$ and
    $g = \reduction{\seq{q_{0,s(\ell)},q_{0,s(\ell+1)} \\
        q_{1,s(\ell)},q_{1,s(\ell+1)} \\ \vdots }}$ for every $\ell$,
    we deduce $q_{m,s(\ell)} = x_k$ for every $\ell$.

    Let $\ell$ be such that $\alpha(\ell)>m$ (it exists since the
    sequence $(\alpha(\ell))_\ell$ is increasing). Then
    $(q_{m,s(\ell)},q_{m,s(\ell+1)}) = (x_k,x_k)$ with
    $x_k \notin \{x_1,\dots,x_{k-1}\}$.  On the other hand, by
    definition of $\alpha(\ell)$,
    $(q_{\alpha(\ell),s(\ell)},q_{\alpha(\ell),s(\ell+1)}) =
    (x_j,y_j)$ and for every $i<\alpha(\ell)$,
    $(q_{i,s(\ell)},q_{i,s(\ell+1)}) \in
    \{(x_1,y_1),\dots,(x_{j-1},y_{j-1})\}$. This yields a
    contradiction since $j<k$.
  \end{itemize}

    We conclude that $g$ is $\omega$-iterable.
    
    Lastly, we prove that $f$ and $g$ are in $\psi(0^+)$. Let $k \in S$, by definition we have 
    \(
    f = f_k = \reduction{\seq{q_{0,0}, \delta(q_{0,0}, \mu_0[:k]) \\ q_{1,0}, \delta(q_{1,0}, \mu_1[:k])\\ \vdots}}.
    \)
    Since $\sigma$ is a winning strategy, there is an index $i$ such that 
    $0^{n} \overset{q_{0,0} \cdot \mu_0[:k]}{\rightsquigarrow} \ldots
    \overset{q_{i,0} \cdot \mu_{i}[:k]}{\rightsquigarrow} 1^n$.
    By Lemma~\ref{lem:morphism-describes-slices-of-winning-paths}, we obtain that $f \in \psi(0^+)$. The proof for $g$ is similar.
\end{proof}

\section{Proof of complexity lower bound}
\label{App:LowerBound}
\begin{proposition}
	\label{prop:hardness}
	\reachtogether is \PSPACE-hard.
\end{proposition}

\begin{proof}
	Let $\calM = (S_\calM, \Sigma, \delta_\calM,s_{\text{init}}, F)$ be a DTM and let $n \in \NN$.
	
	Let $\Gamma = \Sigma \cup S_\calM \times \Sigma  \cup \set{\#}$, and let $m = |\Gamma|$. Let $\beta$ be a bijection between $\Gamma$ and $\llbracket 0,m-1 \rrbracket$. 
	Define, for all $\gamma \in \Gamma$, $\phi(\gamma) = 0^{\beta(\gamma)} 1 0^{m-1-\beta(\gamma)}$.
	Since $\calM$ is deterministic, it has a single run from $(s_{\text{init}},b)b^{n-1}$.
	Let $c_0, c_1, \ldots$ be the sequence of configurations of that run. It may be finite or infinite. We define an infinite word $\rho \in \Gamma^\omega$ describing this sequence.
	If the run is infinite, $\rho = c_0 \# c_1 \# c_2 \# \ldots$. Otherwise, if $c_k$ is the last configuration, $\rho = c_0 \# c_1 \#  \ldots c_k (\# c_k)^\omega$.
	
	Since all configurations have length $n$ and are determined by the transitions of the machine, there is a function $R : \Gamma^3 \to \Gamma$ such that for all $i \in \NN_{>0}$, $\rho[i + n + 1] = R(\rho[i-1], \rho[i], \rho[i+1])$. 
	Note that this is the case even when the run is finite.
%
	
	We apply the morphism $\phi$ to $\rho$ to obtain a word $w$ on $\set{0,1}$.
	From now on we use the term \emph{bit} to mean a single letter of $w$, and \emph{position} to mean the sequence of $m$ bits between positions $i m$ and $(i+1)m -1$, for some $i \in \NN$.
	We say that a wins word is well-shaped if every position is either $0^m$ or $\phi(\gamma)$ for some $\gamma$. 
	 
	We construct an automaton that computes $w$ as follows.
	We say that the automaton reads a letter $\gamma$ if it goes through a sequence of transitions with labels $\neutral^{\beta(\gamma)} \bad \neutral^{m-\beta(\gamma)-1}$. 
	Assuming the current wins word is well-shaped, this is only possible if it has $\phi(\gamma)$ at this position. 
	The automaton writes $0^{\beta(\gamma)} 1 0^{m-\beta(\gamma)-1}$ means going through a sequence of transitions with labels $\neutral^{\beta(\gamma)} \good \neutral^{m-\beta(\gamma)-1}$.
	If the position was $0^m$ or $\phi(\gamma)$ before, it is $\phi(\gamma)$ afterwards.
	The automaton skips a position if it goes through a sequence of $m$ transitions labelled $\neutral$.
	We say that the automaton writes $win$ (resp. reads $win$) on a position if it goes through $m$ transitions labelled $\good$ (resp. $\bad$). 
	
		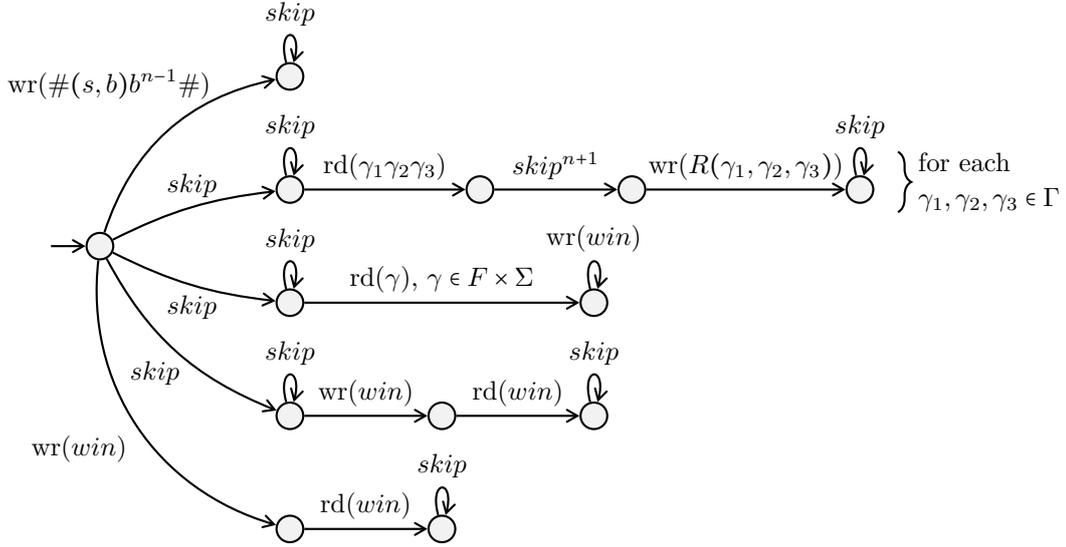
\begin{figure}[ht]
		\begin{tikzpicture}[AUT style]
	

	\node[state, initial] (A) at (0.5,1.75) {};

	\node[state] (init1) at (3, 4)  {};
	
	\draw[->, loop above] (init1) edge node {$skip$} ();
	\draw[->, bend left] (A) edge node[above left, xshift =7mm, yshift=3mm] {wr($\# (s,b) b^{n-1} \#$)} (init1);
	
	\node[state] (trans1) at (3, 2.5)  {};
	\node[state] (trans2) at (5.5, 2.5)  {};
	\node[state] (trans3) at (7.5, 2.5)  {};
	\node[state] (trans4) at (10.5, 2.5)  {};
	
	\draw[->, bend left=10] (A) edge node[above] {$skip$} (trans1);
	\draw[->, loop above] (trans1) edge node {$skip$} ();
	\draw[->] (trans1) edge node[above] {rd($\gamma_1\gamma_2\gamma_3$)} (trans2);
	\draw[->] (trans2) edge node[above] {$skip^{n+1}$} (trans3);
	\draw[->] (trans3) edge node[above] {wr($R(\gamma_1,\gamma_2,\gamma_3)$)} (trans4);
	\draw[->, loop above] (trans4) edge node {$skip$} ();

	\draw [decorate,decoration={brace,amplitude=5pt,mirror}]
	(11,2.2) -- (11,3) node[midway,yshift=-3em]{};
	\node[align=left] (T) at (12.2,2.6) {for each\\ $\gamma_1,\gamma_2,\gamma_3 \in \Gamma$};
	
	\node[state] (win1) at (3, 1)  {};
	\node[state] (win2) at (7, 1)  {};
	
	\draw[->, loop above] (win1) edge node {$skip$} ();
	\draw[->, bend right=10] (A) edge node[below] {$skip$} (win1);
	\draw[->] (win1) edge node[above] {rd($\gamma$), $\gamma \in F \times \Sigma$} (win2);
	\draw[->, loop above] (win2) edge node {wr($win$)} ();
	
	\node[state] (prop1) at (3, -0.5)  {};
	\node[state] (prop2) at (5, -0.5)  {};
	\node[state] (prop3) at (7, -0.5)  {};
	
	\draw[->, loop above] (prop1) edge node {$skip$} ();
	\draw[->, bend right=20] (A) edge node[below left, xshift=5pt] {$skip$} (prop1);
	\draw[->] (prop1) edge node[above] {wr($win$)} (prop2);
	\draw[->] (prop2) edge node[above] {rd($win$)} (prop3);
	\draw[->, loop above] (prop3) edge node {$skip$} ();
	
	\node[state] (tech1) at (3, -2)  {};
	\node[state] (tech2) at (5, -2)  {};
	
	\draw[->, bend right = 40] (A) edge node[below left] {wr($win$)} (tech1);
	\draw[->] (tech1) edge node[above] {rd($win$)} (tech2);
	\draw[->, loop above] (tech2) edge node {$skip$} ();


\end{tikzpicture}
		\caption{The machine for the \PSPACE-hardness reduction. Here $wr$ and $rd$ describe writing and reading actions, while $skip$ (resp. $skip^{n+1}$) stands for a sequence of $m$ $\neutral$ transitions (resp. $(n+1)m$).}
		\label{fig:PSPACE}
	\end{figure}
	
	The automaton can do four things, illustrated in Figure~\ref{fig:PSPACE}:
	\begin{itemize}
		\item It can write $\# (s,b) b^{n-1} \#$ (first branch in the figure)
		
		\item For each $\gamma_1, \gamma_2, \gamma_3 \in \Gamma$, it can read $\gamma_1 \gamma_2 \gamma_3$ at some point in the word, skip $n+1$ positions and write $R(\gamma_1, \gamma_2, \gamma_3)$ (second branch in the figure). 
		
		\item It can read a letter of $F \times \Sigma$ and write $win$ on all following positions (third branch in the figure).
		
		\item It can write $win$ and read $win$ on the next position (two last branches in the figure: the fifth branch lets us apply this on the first position and the fourth branch on the others positions).  
	\end{itemize}

	While the last two items are not applied, all that we can do is apply the first and second items.
	It is clear from the construction that the resulting wins word stays well-shaped.
	Further, the resulting word will always be less or equal to $w$ for the $\preceq$ partial ordering: this results from the determinism of the machine: for every position there is only one letter that we can write on it.
	
	If $w$ contains a letter of $F \times \Sigma$ at some position, we will eventually write it, and apply the third item to set all following bits to $1$, and the last item repeatedly for the rest of the bits.
	
	Otherwise, we will only obtain wins words with $\phi(\#)$ on their first position, and thus some of the $m$ first bits will stay $0$ forever.  
	Hence there is a winning strategy if and only if the run of the Turing machine reaches a final state.
\end{proof}

\end{document}